\documentclass[12pt]{article}

\usepackage{amsmath}
\usepackage{epsfig, amssymb, amsfonts}
\usepackage{amsthm}
\usepackage{amstext}
\usepackage{amsopn}
\usepackage{enumerate}
\usepackage{mathrsfs}
\usepackage{subfigure}
\usepackage{graphicx}
\usepackage[left=2.3cm,top=2cm,right=2.3cm,bottom=2cm, nohead,foot=1cm]{geometry}
\numberwithin{equation}{section}

\newtheorem{theorem}{Theorem}[section]
\newtheorem{lemma}[theorem]{Lemma}
\newtheorem{assumptions}[theorem]{List}
\newtheorem{remark}[theorem]{Remarks}

\newtheorem{proposition}[theorem]{Proposition}
\newtheorem{definition}[theorem]{Definition}

\newtheorem{example}[theorem]{Example}

\newcommand{\R}{{\mathbb R}}

\newcommand{\mcJ}{{\mathcal{J}}}
\newcommand{\tuP}{{\Psi}}

\newcommand{\Z}{{\mathbb Z}}

\title{Bounds for the state-modulated resolvent of a linear Boltzmann generator}

\author{\textbf{Jeremy Clark}\footnote{jclark@mappi.helsinki.fi} \quad   and \quad \textbf{Lo\"ic Dubois}\footnote{ldubois@mappi.helsinki.fi} \\ University of Helsinki, Department of Mathematics\\  Helsinki 00014, Finland }

\begin{document}
\maketitle

\begin{abstract}

We study a generalized resolvent for the generator of a Markovian semigroup.  The Markovian generator appears in a linear Boltzmann equation modeling a one-dimensional test particle in a periodic potential and colliding elastically with particles from an ideal background gas.  We obtain bounds for the state-modulated resolvent which are relevant in the regime where the mass ratio between the test particle and a particle from the gas is large.  These bounds relate to the typical amount of time that the particle spends in different regions of phase space before arriving to a region around the origin.  
    
\end{abstract}

\section{Introduction}

\subsection{Model and result}
Denote  $\Sigma=\mathbb{T}\times \R$, where  $\mathbb{T}$ is the one-dimensional torus identified with the unit interval $[0,1)$.  Let  $B(\Sigma)$ be the Banach space of all bounded measurable functions on $\Sigma$ with  the supremum norm $\|\cdot\|_{\infty}$.  Let $\mathcal{L}_{\lambda}$  be the backward Kolmogorov generator which acts on  
a dense domain $D\subset B(\Sigma)$  such that for $ \tuP\in D$,
\begin{multline}\label{TheModel}
(\mathcal{L}_{\lambda}\tuP)(x,\,p)=p\frac{\partial}{\partial x}\tuP(x,p)-\frac{dV}{dx}\big(x\big)\frac{\partial}{\partial p}\tuP(x,p)     +\int_{\R}dp^{\prime}\mcJ_{\lambda}(p,p^{\prime})\big(\tuP(x,p^{\prime})-\tuP(x,p)     \big),   
\end{multline}
 where the function $V:\mathbb{T}\rightarrow \R^{+}$ is continuously differentiable, and the kernel $\mcJ_{\lambda}(p,p^{\prime})$ takes the form 
\begin{align}\label{JumpRates}
  \mcJ_{\lambda}(p,p^{\prime})=  (1+\lambda)\big|p-p^{\prime}\big|e^{-\frac{1}{2}\big(\frac{1-\lambda}{2}p -\frac{1+\lambda}{2}p^{\prime}    \big)^{2} }     .   
  \end{align}  
The operator  $\mathcal{L}_{\lambda}$ generates a transition semigroup
 $\Phi_{t,\lambda}:B(\Sigma)\rightarrow B(\Sigma)$.   The Markovian dynamics associated with the semigroup $\Phi_{t,\lambda}$ models a test particle of mass $\lambda^{-1}$ in dimension one which feels an external, spatially periodic force $\frac{dV}{dx}(x)$ and receives elastic collisions from a gas reservoir of particles having mass one. The spatial degree of freedom for the test particle has been contracted to the torus.   The kernel~(\ref{JumpRates}) matches equation (8.118) from~\cite{Spohn} when the test particle has mass $\lambda^{-1}$, a single particle from the gas has mass one, the temperature of the gas is one, and the spatial density for the gas is $2(2\pi)^{\frac{1}{2}}  $.

Consider the generalized resolvent $U^{(\lambda)}_{h}$ which operates on elements in $B(\Sigma)$ and is given formally by
\begin{align}\label{ResolveIt} 
 U^{(\lambda)}_{h}=(M_{h}-\mathcal{L}_{\lambda})^{-1},     
 \end{align}
where  $M_{h}:B(\Sigma)\rightarrow B(\Sigma)$ acts as multiplication by a bounded measurable function $h:\Sigma\rightarrow\R^{+}$.   When  $h$ is a constant function, then  $U^{(\lambda)}_{h}$ is a standard resolvent.   We will refer to  $U^{(\lambda)}_{h}$ as the $h$-\textit{modulated resolvent} or,  as in~\cite{Meyn},  the \textit{state-modulated resolvent} for non-specific $h$.  The operators  $U^{(\lambda)}_{h}$ were introduced in~\cite[Sec. 7]{Neveu}.  For $s\in \Sigma$ and $f\in B(\Sigma)$, we will use the kernel notation $U^{(\lambda)}_{h}(s,f)$ to denote the value  $\big(U^{(\lambda)}_{h}\,f\big)(s)$.  We take the following expression for our definition of $U^{(\lambda)}_{h}$:
\begin{align*}
\hspace{2cm} U^{(\lambda)}_{h}\big(s,f\big)=\mathbb{E}_{s}^{(\lambda)}\Big[\int_{0}^{\infty}dt\, f(S_{t})\,e^{-\int_{0}^{t}dr\,h(S_{r} )}    \Big],\hspace{1cm} s\in \Sigma,
\end{align*}
where $S_{t}\in \Sigma $ is the Markov process associated with the backwards generator $\mathcal{L}_{\lambda}$.   The operator $U^{(\lambda)}_{h}$ satisfies~(\ref{ResolveIt}) on an appropriate class of functions $f\in B(\Sigma)$.  For a measurable set $A\subset \Sigma $, the value  $U^{(\lambda)}_{h}(s,1_{A})$   corresponds to the expected amount of time, when starting from $s$, that the test particle will spend in $A \subset \Sigma $ before the expiration of an exponential random time whose rate depends on the trajectory of the particle through the function $h$.  This interpretation becomes clearer by seeing other representations of  $U^{(\lambda)}_{h}$.  The following theorem is used in~\cite{Previous} and is the main result of this article.    

\begin{theorem}\label{ThmMain}\text{  }
  Let  $h:\Sigma\rightarrow \R^{+}$ be a bounded measurable function with $h\neq 0$.  There is a $c>0$ such that for all bounded measurable functions $f:\Sigma\rightarrow \R^{+}$,  $\lambda< 1$, and $s\in \Sigma$
$$U^{(\lambda)}_{h}\big(s,f\big)\leq c\Big(\sup_{s'\in \Sigma} A^{(\lambda)}(s,s')f(s')+\int_{\Sigma}ds'\, B^{(\lambda)}(s,s')f(s')\Big),  $$
 where $ A^{(\lambda)}(s,s')$ and $  B^{(\lambda)}(s,s')$ are defined as
\begin{align*}
A^{(\lambda)}(x,p,\,x',p')&= 1+\textup{min}\big(|p|, \lambda^{-1} \big)\chi\big( |p'|\geq \lambda^{-1}\big),\\    
B^{(\lambda)}(x,p,\,x',p')&= \big(1+\textup{min}\big(|p|, |p'| \big) \big) \chi( |p|\leq \lambda^{-1}).
 \end{align*} 
  
\end{theorem}

\subsection{Discussion}

  The  operator $U^{(\lambda)}_{h}$ arises in the literature on Harris recurrence for Markov processes~\cite{Neveu,Meyn,Huis}, on limit theorems for null-recurrent Markov processes~\cite{TouatiUnpub,Hopfner,Louk}, and on Nummelin splitting for Markov processes~\cite{Loch}.   We discuss some alternative representations for $U^{(\lambda)}_{h}f $ in Section~\ref{SecBasic}.

The dynamics described by~(\ref{TheModel}) includes a deterministic part driven by the Hamiltonian $H=\frac{1}{2}P^{2}+V(x)$ and a noisy part determined by the jump kernel $\mathcal{J}_{\lambda}$.  The potential $V(x)$  does not play a role in the statement of Theorem~\ref{ThmMain}.   This is because the inequality in Theorem~\ref{ThmMain} is mainly concerned with bounding $U_{h}^{(\lambda)}(x,p)$ when $|p|\gg 1$, and the influence of the force $\frac{dV}{dx}(X_{t})$ is averaged-out when  $|P_{t}|\gg 1+ \sqrt{\sup_{x} V(x)} $ as the particle revolves with high frequency around the torus.  In whichever direction the test particle is traveling, collisions with the gas will, in an average sense, diminish the particle's movement in that direction.   For $\lambda\ll 1$, this frictional effect takes on different characteristics depending on the scale of the momentum.   The following list  characterizes the influence of collisions  at different momentum scales relative to $\lambda^{-1}$.

\begin{assumptions}\label{RegimeList} \text{  }

\begin{enumerate}

\item \textbf{Contractive regime}: When $|p|\gg \lambda^{-1}$, the momentum undergoes a super-exponential contraction in which the collisions occur with exponential rate $\approx \lambda|p|$, and the result of a  collision  contracts a momentum $p$ to a value in the vicinity of $\frac{1-\lambda}{1+\lambda}p$.

\item  \textbf{Drift regime}: When $|p|$ is on the order $\lambda^{-1}$, then the collisions occur on the order of one per unit time and the resulting momentum due to a collision has a bias in the direction of momentum zero.  The bias has the same order as the standard deviation of the momentum jump.

\item  \textbf{Random walk regime}: When $|p|\ll \lambda^{-1}$, then the collisions generate a nearly unbiased random walk in momentum with L\'evy density
\begin{align}\label{JumpRatesUnbiased}
  j(q)=  |q|e^{-\frac{1}{8}q^{2} }     .   
  \end{align}  
Over time periods of length $\gg 1$, then the drift towards momentum zero is visible. 

\end{enumerate}

\end{assumptions}

Besides the multi-scale behavior for the jump rates~(\ref{RegimeList}) emerging for $\lambda\ll 1$, 
the main source of technical difficulty for proving Theorem~\ref{ThmMain} is the perturbation of the dynamics due to the presence of the potential $V(x)$.  Since the potential is bounded, the particle revolves quickly around the torus at high energies.  This gives rise to an effective homogenized behavior at high-energy which is quasi $1$-dimensional.  This homogenized dynamics 
can be formulated as a Markov process which has states that can be identified the with connected components of the level curves for the Hamiltonian $H=\frac{1}{2}p^{2}+V(x)$.  We refer to the homogenized process as the \textit{Freidlin-Wentzell process}, since it is similar to processes which arise in Freidlin-Wentzell limits~\cite{Wentzell}, although it is not of a diffusive form.  The Freidlin-Wentzell process is more tractable, since there is no drift between collisions and the torus degree of freedom is replaced by a finite labeling. Our strategy for handling the potential is to prove an analog of Theorem~\ref{ThmMain} for the corresponding Freidlin-Wentzell process (see Lemma~\ref{Freidlin}), and then to show in Lemma~\ref{ErrorEquation} that the state-modulated resolvent for the original process satisfies the same integral equation as the state-modulated resolvent for the Freidlin-Wentzell process except for an error that can be controlled.      
 Although the behavior of the original dynamics and the homogenized dynamics diverges at low energy, the cumulative effect of the divergence for the state-modulated resolvent still conforms to the bounds that we consider in Theorem~\ref{ThmMain}.  
 
 The bounds in Theorem~\ref{ThmMain} are not optimal.  Our analysis does not take advantage of the drift described in (2) of List~\ref{RegimeList}, which should allow for tighter bounds.  There are also smaller kernels than $A^{(\lambda)}$ and $B^{(\lambda)}$ possible over the domain of $(x,p,\,x',p')$ where $p$ and $p'$ have opposite signs\footnote{This can be understood from the Example~\ref{ExpBrownRes}.}, although we are not interested in this here.  
 Finally,  by a slightly different analysis of the contractive regime (3) of List~\ref{RegimeList}, the kernel $A^{(\lambda)}$ can be replaced by the kernel $A^{(\lambda),\prime}$ defined as 
$$ A^{(\lambda),\prime}(x,p,\,x',p')= \Big(1+\textup{min}\big(|p|, \lambda^{-1}\log(1+\lambda p) \big)\chi\big( |p'|\geq \lambda^{-1}\big)\Big)
  \frac{1}{1+\lambda|p'|}.  $$
 This alternative is not strictly stronger than the choice of $A^{(\lambda)}$.
 
The remainder of the article is arranged as follows:  In the next Section, we give a few examples of inequalities for the state-modulated resolvent for simpler processes.  Section~\ref{SecBasic} contains some general remarks on state-modulated resolvents and also contains some technical preliminaries specific to our dynamics.   Section~\ref{SecFreidlin} discusses the Freidlin-Wentzell process, and Section~\ref{SecTheLink}  bridges the analysis of the Freidlin-Wentzell process with the original process.  The proof of Theorem~\ref{ThmMain} is in Section~\ref{SecTheProof}.

\subsection{Examples of inequalities for state-modulated resolvents}

The bound for $U^{(\lambda)}_{h} $ given in Theorem~\ref{ThmMain} is especially complicated due to the different scales described in List~\ref{RegimeList}.  In particular, the inequality in Theorem~\ref{ThmMain} involves two kernels $A^{(\lambda)}(s,s')$ and $B^{(\lambda)}(s,s')$ which are used in supremum and integral norms, respectively.  The first two examples below only involve random walk behavior (i.e. (3) of List~\ref{RegimeList}), and the kernel  $B^{(\lambda)}(s,s')$ is sufficient. Examples~\ref{LevyProcess} and~\ref{WithoutTorus} follow by simpler analysis than contained in Section~\ref{SecFreidlin}.

\begin{example}\label{ExpBrownRes}
Let  $\mathbf{B}$ be a one-dimensional standard Brownian motion and $f:\R\rightarrow \R^{+}$ be integrable. Define
the kernel $U_{\frak{h}}$ such that for $p\in \R$, 
$$ U_{\frak{h}}\big(p, f\big)= \mathbb{E}_{p}\Big[\int_{0}^{\infty}dt\, f(\mathbf{B}_{t})e^{-\frak{h}\frak{l}_{t}}\Big],       $$
where $\frak{h}>0$ and $\frak{l}_{t}$ is the local time at zero of $\mathbf{B}$.  In words, $U_{\frak{h}}$ is the $h$-modulated resolvent of $f$, where $h$ is the $\delta$-function $h(p)=\frak{h}\delta_{0}(p)$.   The function $U_{\frak{h}} f$ satisfies the differential equation    
$$ f(p)=\frak{h}\,\delta_{0}(p)\,U_{\frak{h}}\big(p, f\big) -\frac{1}{2}\Delta_{p}\,U\big(p, f\big).  $$
A closed form of the solution is given by
$$ U_{\frak{h}}\big(p,f\big)=\frac{1}{\frak{h}}\int_{\R}f(p)dp+  2 \left\{  \begin{array}{cc} \int_{0}^{\infty}dq\,\,\textup{min}\big(q, p\big)\,f(q) &  p\geq0,   \\  \quad & \quad \\  \int_{0}^{\infty}dq\,\textup{min}\big(q, -p\big)\,f(-q)   & p<0 .   \end{array} \right.     $$
  Trivially, there exists a $c>0$ such that for $B(p,p'):=1+\textup{min}\big(|p'|, |p| \big) $ and all $p\in \R$ and integrable $f:\
\R\rightarrow \R^{+}$,
$$\hspace{2cm} U_{\frak{h}}(p,f)\leq  c\int_{-\infty}^{\infty}dq\,B(p,p')\,f(p').    $$

\end{example}

\begin{example}\label{LevyProcess}
Let the backward Kolmogorov generator $\mathcal{L}$ be defined such that for $\tuP\in B(\R)$,
\begin{align*}
(\mathcal{L}\tuP)(p)=\int_{\R}dp^{\prime}j(p^{\prime}-p)\big(\tuP(p^{\prime})-\tuP(p)     \big),   
\end{align*}
where  $j:\R\rightarrow \R^{+}$ is integrable, continuous, and its first two moments satisfy $\int_{\R}dp\,p\,j(p)=0$ and $\int_{\R} dp\,p^{2}\,j(p)<\infty$.  
For measurable $h:\R\rightarrow \R^{+}$ with $h\neq 0$, there is a $c>0$ such that
for $B(p',p)$ defined as in example~(\ref{ExpBrownRes}) and all $p\in \R$ and integrable $f:\
\R\rightarrow \R^{+}$,
\begin{eqnarray*}
U_{h}\big(p,f\big) \leq c\int_{\R }dp'\,B(p',p)f(p').    
\end{eqnarray*}

\end{example}

\begin{example}\label{WithoutTorus}
Consider the backwards Markov generator $\mathcal{L}_{\lambda}$ which acts on a dense domain of $B(\R)$ as
$$
(\mathcal{L}_{\lambda}\tuP)(p)= \int_{\R}dp^{\prime} \mathcal{J}_{\lambda}(p,p^{\prime})\big(\tuP(p^{\prime})-\tuP(p)     \big), 
$$
where $\mathcal{J}_{\lambda}$ is defined as in~(\ref{JumpRates}).   Let  $h:\R\rightarrow \R^{+}$ be measurable and $h\neq 0$.  There is a $c>0$ such that for all $f\in B(\R)$ with  $f\geq 0$,  $\lambda< 1$, and $p\in \R$,
$$U^{(\lambda)}_{h}\big(p,f\big)\leq c\Big(\sup_{p'\in \Sigma} A^{(\lambda)}(p,p')f(p')+\int_{\Sigma}ds'\, B^{(\lambda)}(p,p')f(p')\Big),  $$
 where $ A^{(\lambda)}(p,p')$ and $  B^{(\lambda)}(p,p')$ are defined as
\begin{align*}
A^{(\lambda)}(p,p')&= 1+\textup{min}\big(|p|, \lambda^{-1}   \big)\chi\big( |p'|\geq \lambda^{-1}\big),\\    
B^{(\lambda)}(p,p')&= \big(1+\textup{min}\big(|p|, |p'| \big) \big)\chi( |p'|\leq \lambda^{-1}).
 \end{align*}

\end{example}

\section{Some basic facts for the state-modulated resolvent}\label{SecBasic}

   Proposition~\ref{LifeOperator} gives alternative representations for $U^{(\lambda)}_{h}\big(s,f\big)$ where (2) is from~\cite[Sec. 7]{Neveu}, (1) is from~\cite[Sec. 2]{Meyn}, and (3) is from the proof of~\cite[Prop. 3.4]{Hopfner} (for $\mathbf{h}=1$).   
   
\begin{proposition}\label{LifeOperator}
 Let  $f,h\in B(\Sigma)$, where $h$ is non-negative and $h\neq 0$.
 Pick $\mathbf{h}\geq \sup_{s\in \Sigma}h(s)$.  
 
\begin{enumerate}

\item Let $R$ be the stopping time with infinitesimal exponential rate at a time $t<R$ given by $h(S_{t})$, i.e. for all $t\in \R^{+}$ and $\delta\ll 1$
$$  \mathbb{P}_{s}^{(\lambda)}\big[ R\in [t,t+\delta]\,  \big|\,R\geq t,\, S_{r} \text{ for }  r\in [0,t]       \big]=h(S_{t})\delta+\mathit{o}(\delta).    $$
The function $U^{(\lambda)}_{h}f$ can be written as 
$$ U^{(\lambda)}_{h}\big(s,f\big)= \mathbb{E}_{s}^{(\lambda)}\Big[\int_{0}^{R}dr\,f(S_{r})  \Big].   $$

\item   Let $M_{h'}:B(\Sigma)\rightarrow B(\Sigma)$ be multiplication by $h'(s)= \mathbf{h}-h(s)$ and $U^{(\lambda)}_{\mathbf{h}}$ be the standard resolvent evaluated at $\mathbf{h}\in \R^{+}$.    The operator $U^{(\lambda)}_{h}$  can be written as
$$ U^{(\lambda)}_{h}=\sum_{n=0}^{\infty}  U^{(\lambda)}_{\mathbf{h}} \big( M_{h'}\,   U^{(\lambda)}_{\mathbf{h}} \big)^{n}.     $$

\item Let $e_{n}$ be independent, mean-$\mathbf{h}^{-1}$ exponentials which are independent of $S_{t}$.   For $\tau_{n}=\sum_{m=1}^{n}e_{m}$, the function $U^{(\lambda)}_{h}f$ can be written as
$$ U^{(\lambda)}_{h}\big(s,f\big)=\mathbb{E}_{s}^{(\lambda)}\Big[\sum_{n=1}^{\infty}\,  \Big(1-\frac{h(S_{\tau_{1}})}{\mathbf{h}}\Big)\cdots \Big(1-\frac{h(S_{\tau_{n-1}})}{\mathbf{h}}\Big)f(S_{\tau_{n}}) \Big]. $$

\item  If coins with head weight $\frac{h(S_{\tau_{n}})}{\mathbf{h}}$ are flipped at every time $\tau_{n}$ and $\tilde{n}\in \mathbb{N}$ is the count of the first head, then $U^{(\lambda)}_{h}f$ can be written as
$$ U^{(\lambda)}_{h}\big(s,f\big)=\mathbb{E}_{s}^{(\lambda)}\Big[\sum_{n=1}^{\tilde{n}}\, f(S_{\tau_{n}}) \Big].  $$

\end{enumerate}

\end{proposition}

The following lemma is specific to our dynamics, and Part (3) implies that it is sufficient for us to prove Theorem~\ref{ThmMain} for a function $h:\Sigma\rightarrow \R^{+}$ of our choosing as long as it has compact support.     In later sections of this article, we will always take 
\begin{align}\label{OurAitch}
 h(s)= \chi\big( H(s)\leq l    \big),    
 \end{align}
where $l= 1+2\sup_{x}V(x)$.  We pick $l$ primarily to ensure that the particle is not trapped by the potential when $H(S_{t})>l$ and is revolving around the torus with speed $>1$ over the time period up to the next collision.

For the proof of (1) from  Lemma~\ref{LemPick} below, we use that our dynamics is exponentially ergodic, which was proven in~\cite[Thm. A.1]{Previous}.   Let $\mathcal{T}_{\lambda,\frac{1}{\mathbf{h}}}:B(\Sigma)\rightarrow B(\Sigma)$ be the operator
\begin{align*}
\mathcal{T}_{\lambda,\frac{1}{\mathbf{h}}}:= \mathbf{h}\int_{0}^{\infty}dt\,e^{-t\mathbf{h} +t\mathcal{L}_{\lambda} }=\mathbf{h}U_{\mathbf{h}}^{(\lambda)}.
\end{align*}
 The kernel $ \mathcal{T}_{\lambda, \frac{1}{\mathbf{h}}}(s,ds')$ is the transition kernel for the resolvent chain $S_{\tau_{n}}$ from (3) of Proposition~\ref{LifeOperator}.  To prove (2) and (3) of Lemma~\ref{LemPick}, we use that there is a $c_{L,\mathbf{h}}>0$ such that the forward transition operator $\mathcal{T}_{\lambda,\frac{1}{\mathbf{h}}}$ satisfies
  \begin{align}\label{PreHormander}
  \mathcal{T}_{\lambda,\frac{1}{\mathbf{h}}}(s,ds') \geq c_{L,\mathbf{h}} \,ds'    
  \end{align}
for all $s,s'\in \Sigma$ with $H(s),H(s')\leq L$ and all $\lambda<1$. 
This was shown in the proof of~\cite[Prop. 4.3]{Previous}\footnote{In the proof of~\cite[Thm. A.1]{Previous}, $\mathbf{h}=1$ and $L=1+2\sup_{x}V(x)$, although this does not matter for the argument.}.

\begin{lemma}\label{LemPick} Let $h,h',f\in B(\Sigma)$ be non-negative and $h,h'\neq 0 $. 

\begin{enumerate}

\item   The kernel $U^{(\lambda)}_{h} $ defines a bounded map on $B(\Sigma)$ (i.e. with respect to the supremum norm).

\item  Let $h$ be as in~(\ref{OurAitch}) and pick $L>0$.  There is a $c_{L}>0$ such that for all $f$ and $\lambda<1$,
  \begin{align*}  
  \Big| \int_{H(s)\leq L}ds\, U^{(\lambda)}_{h}\big(s,f\big)\Big| &\leq  c_{L}\int_{\Sigma}ds\, f(s)\,e^{-\lambda H(s)}\\     
& \leq 2c_{L}\sup_{H(s)> \frac{1}{2}\lambda^{-2} } f(s)+c_{L}\int_{ H(s)\leq \frac{1}{2}\lambda^{-2} }ds\,f(s).
\end{align*}

\item  Suppose that $h'$ has compact support. There are $c,L>0$ 
such that for all $f$ and $\lambda<1$,
$$  U^{(\lambda)}_{h}f\leq c\sup_{H(s)\leq L}U^{(\lambda)}_{h'} \big(s,f\big) +\, U^{(\lambda)}_{h'}f  .  $$

\end{enumerate}

\end{lemma}

\begin{proof} \text{  }\\
\noindent Part (1):   Recall that  $\mathcal{T}_{\lambda, \frac{1}{\mathbf{h}}}= \mathbf{h}U^{(\lambda)}_{\mathbf{h}} $.  By Part (2) of Proposition~\ref{LifeOperator}, we can write
\begin{align}\label{LadyBugs}
U^{(\lambda)}_{h}=  \frac{1}{\mathbf{h}} \sum_{n=0}^{\infty}    \mathcal{T}_{\lambda,\frac{1}{\mathbf{h}}}\big(M_{\frac{\mathbf{h}-h}{\mathbf{h}} }\mathcal{T}_{\lambda,\frac{1}{\mathbf{h}}}\big)^{n}. 
\end{align}

Let $\Phi_{\lambda,t}:B(\Sigma)\rightarrow B(\Sigma)$ be the transition semigroup associated to the backward Kolmogorov equation~(\ref{TheModel}).   
By~\cite[Thm. A.1]{Previous}, $\Phi_{\lambda,t}$ converges exponentially  in the operator norm  to the equilibrium projection $\mathbf{P}_{\lambda}=1_{\Sigma} \otimes \Psi_{\infty,\lambda}$ as $t\rightarrow \infty$.   It follows that operators $\mathcal{T}_{\lambda,\frac{1}{\mathbf{h}} }^{n}$ also converge exponentially to $\mathbf{P}_{\lambda}$ for large $n$.  In other terms, there are $C,\alpha>0$ such that for all $f\in B(\Sigma)$,
$$\| \mathcal{T}_{\lambda,\frac{1}{\mathbf{h}}}^{n}f-\mathbf{P}_{\lambda}f\|_{\infty}= \| \mathcal{T}_{\lambda,\frac{1}{\mathbf{h}}}^{n}f-\big(\Psi_{\infty,\lambda}(f)\big)1_{\Sigma}\|_{\infty}\leq C\,e^{-n\alpha }\|f\|_{\infty}.   $$
Pick an $N>0$ such that $C\,e^{-N\alpha}\leq \frac{1}{2}\Psi_{\infty,\lambda}(h)$. Using the form~(\ref{LadyBugs}) and that $M_{\frac{\mathbf{h}-h}{\mathbf{h}}}$ is a positive  multiplication operator with norm $\leq 1$,    we get the first inequality below:
\begin{align}
\| U^{(\lambda)}_{h}f\|_{\infty}&\leq \frac{ N}{ \mathbf{h} } \|f\|_{\infty} \Big\|\sum_{n=0}^{\infty} \big(\mathcal{T}_{\lambda,\frac{1}{\mathbf{h}}}^{N}M_{\frac{\mathbf{h}-h}{\mathbf{h}} }\big)^{n} \Big\| \nonumber \\ & \leq \frac{N}{ \mathbf{h} } \|f\|_{\infty} \sum_{n=0}^{\infty} \big\|\mathcal{T}_{\lambda,\frac{1}{\mathbf{h}}}^{N}\frac{\mathbf{h}-h}{\mathbf{h}} \big\|_{\infty}^{n}   \leq  \frac{2N  \|f\|_{\infty}}{ \mathbf{h}\Psi_{\infty,\lambda}(h)  } ,
\end{align}
where $\|\cdot\|$ denotes the operator norm.

The third inequality uses that
$$0\leq  1- \big\|\mathcal{T}_{\lambda,\frac{1}{\mathbf{h}} }^{N}\frac{
\mathbf{h}-h}{\mathbf{h}} \big\|_{\infty}\leq 1-\frac{1}{\mathbf{h}}\| \mathbf{P}_{\lambda}h\|_{\infty}+ \frac{1}{2\mathbf{h}}\Psi_{\infty,\lambda}(h)=1- \frac{1}{2\mathbf{h}}\Psi_{\infty,\lambda}(h).  $$
Thus, $U^{(\lambda)}_{h}$ is a bounded operator.

\vspace{.5cm}

\noindent Part (2):  Let $ \nu(ds)=\frac{\mu(ds)}{\mu(R)   }1_{R}   $ be the normalization of Lebesgue measure $\mu$ over the set $R=\{s\in \Sigma\,\big|\,H(s)\leq l   \}$ and $h'':\Sigma\rightarrow \R^{+}$ be the function $h''= c_{l,\frac{1}{2}}\mu(R)1_{R}$.   By~(\ref{PreHormander}), the transition kernel for  $\mathcal{T}_{\lambda,\frac{1}{\mathbf{h}}}$ satisfies
  \begin{align}\label{Hormander}
  \mathcal{T}_{\lambda,\frac{1}{\mathbf{h}}}(s,ds')\geq h''(s)\nu(ds')   
  \end{align}
  for all $s,s'\in \Sigma$ with $H(s),H(s')\leq l$.  Since $h\leq 1$ and by~(\ref{Hormander}),
$$M_{\sqrt{1-\frac{h}{2} }}\mathcal{T}_{\lambda,\frac{1}{2} }M_{\sqrt{1-\frac{h}{2}}}\leq \mathcal{T}_{\lambda,\frac{1}{2} }- M_{1-\sqrt{1-\frac{h}{2} }}\,\mathcal{T}_{\lambda,\frac{1}{2} } \,M_{1-\sqrt{1-\frac{h}{2}   }}\leq \mathcal{T}_{\lambda,\frac{1}{2} }-\frac{1}{10} h''\otimes \nu  . $$
With (\ref{LadyBugs}) and the fact that $\sqrt{1-\frac{h}{2}}\geq 2^{-\frac{1}{2}}$,  
  $$  
U^{(\lambda)}_{h}\leq 2M_{\sqrt{1-\frac{h}{2} }}U^{(\lambda)}_{h} M_{\sqrt{1-\frac{h}{2} }}      =2\sum_{n=0}^{\infty} \big(M_{\sqrt{1-\frac{h}{2} }}\mathcal{T}_{\lambda,\frac{1}{2} }M_{\sqrt{1-\frac{h}{2} }}\big)^{n}\leq    2\sum_{n=0}^{\infty} \big(\mathcal{T}_{\lambda,\frac{1}{2} } -\frac{1}{10}h''\otimes \nu     \big)^{n} .
$$
Thus, we have the inequality below 
$$\nu\,  U^{(\lambda)}_{h}f\leq     2\,\nu   \sum_{n=0}^{\infty} \big(\mathcal{T}_{\lambda,\frac{1}{2} } -\frac{1}{10}h''\otimes \nu     \big)^{n}f =   20 \frac{ \int_{\Sigma}ds\,\Psi_{\infty,\lambda}(s) \,f(s)       }{  \int_{\Sigma}ds\,\Psi_{\infty,\lambda}(s) \,h''(s) } =  20 \frac{ \int_{\Sigma}ds\,e^{-\lambda H(s)} \,f(s)       }{  \int_{\Sigma}ds\,e^{-\lambda H(s)} \,h''(s) } .  $$
The first equality is an identity from~\cite[Thm. 3]{Nummelin}.  However, 
\begin{align}\label{Forfeit}
\nu\,  U^{(\lambda)}_{h}f\geq \frac{1}{2}\,\nu \,\mathcal{T}_{\lambda,\frac{1}{2} }U^{(\lambda)}_{h}\geq \frac{c_{L,\frac{1}{2}}}{2}\int_{H\leq L}ds\,U^{(\lambda)}_{h}\big(s,\,f\big),  
\end{align}
where the second inequality is by~(\ref{PreHormander}).  The first inequality in~(\ref{Forfeit}) follows by the relations 
$$   U^{(\lambda)}_{h}=\mathcal{T}_{\lambda,\frac{1}{2} }+\mathcal{T}_{\lambda,\frac{1}{2} }M_{1-\frac{h}{2}} U^{(\lambda)}_{h}\geq \frac{1}{2} \mathcal{T}_{\lambda,\frac{1}{2} } U^{(\lambda)}_{h},    $$
where the equality is equivalent to Part (2) of Proposition~\ref{LifeOperator} and the inequality is $1-\frac{h}{2}\geq \frac{1}{2}$.  Combining~(\ref{Forfeit}) with the inequality above it gives that $U^{(\lambda)}_{h}f$ is bounded by a multiple of  $ \int_{\Sigma}ds\,e^{-\lambda H(s)} \,f(s) $     .  

Finally, 
\begin{align*}
\int_{\Sigma}ds\,e^{-\lambda H(s)} \,f(s) & \leq    
\Big(\sup_{H> \frac{1}{2}\lambda^{-2} } f(s) \Big) \int_{H> \frac{1}{2}\lambda^{-2} }ds\,e^{-\lambda H(s)}  +\int_{ H\leq \frac{1}{2}\lambda^{-2} }ds\,f(s), \\
&   \leq  2 \sup_{H> \frac{1}{2}\lambda^{-2} } f(s) +\int_{ H\leq \frac{1}{2}\lambda^{-2} }ds\,f(s)
\end{align*}
where we have split the integration into the domains $H'> \frac{1}{2}\lambda^{-2}$ and $H'\leq  \frac{1}{2}\lambda^{-2}$ and the second inequality is for $\lambda$ small enough.

\vspace{.5cm}

\noindent Part (3):     Since $ U^{(\lambda)}_{h_{2} }f\leq \, U^{(\lambda)}_{h_{1}}f$ when $h_{1}\leq h_{2}$, we can assume without loss of generality that $h$ also has compact support.  For the same reason, we can take $h'$ to be of the  form $h'=\mathbf{h}1_{A}$ for some $\mathbf{h}>0$ and where $A=\{s\,\big|\, H(s)\leq L \}$ for some $L>0$.

Define the kernels  $\vartheta_{\lambda}, \varrho_{\lambda}:B(\Sigma)\rightarrow B(\Sigma)$ such that
\begin{align*}
  \vartheta_{\lambda} &=\sum_{m=0}^{\infty} \big(\mathcal{T}_{\lambda,\frac{1}{\mathbf{h}} }  M_{1_{A^{c}}}\big)^{n}\mathcal{T}_{\lambda,\frac{1}{\mathbf{h}} }  M_{1_{A}} ,\\
\varrho_{\lambda}& =  \sum_{n=1}^{\infty} \vartheta_{\lambda} \big(M_{1-\frac{h}{\mathbf{h}} } \vartheta_{\lambda}\big)^{n}.
\end{align*}
For each $s\in \Sigma $, $\vartheta_{\lambda}(s,ds')$ is a probability measure supported on $A$, and $\varrho_{\lambda}(s,ds') $ is a measure with total weight bounded by 
\begin{align}\label{Carp}
\sup_{s\in \Sigma}\varrho_{\lambda}(s,\Sigma)& \leq \sum_{n=1}^{\infty}  \sup_{s\in \Sigma}\,\vartheta_{\lambda} \big(M_{1-\frac{h}{\mathbf{h}} } \vartheta_{\lambda}\big)^{n} \big(s,\Sigma\big) \leq 2+\sum_{n=1}^{\infty} \Big( \sup_{s\in A} \vartheta_{\lambda} M_{1-\frac{h}{\mathbf{h}} } \big(s,\Sigma\big)      \Big)^{n} \noindent   \\ & \leq 2+\sum_{n=1}^{\infty} \Big(1-c_{L,\mathbf{h}}\int_{A}ds'\,\frac{h(s')}{\mathbf{h}}  \Big)^{n} \leq 2+\frac{\mathbf{h} }{c_{L,\mathbf{h}}\int_{A}ds'\,h(s') }:=C, 
\end{align}
where the third inequality uses the remark~(\ref{PreHormander}).  The second inequality above follows since $\varrho_{\lambda}(s,ds')$ is supported in $A$.

The resolvent $U^{(\lambda)}_{h}$ can be written in terms of $U^{(\lambda)}_{h'}$ and $\varrho_{\lambda}$ as
\begin{align*}
U^{(\lambda)}_{h}\big(s,f\big) &=  \int_{\Sigma} \varrho_{\lambda}(s,ds')U^{(\lambda)}_{h'}\big(s',f\big)+ U^{(\lambda)}_{h'}\big(s,f\big) \\
&\leq    \Big(\sup_{s''\in \Sigma}\int_{\Sigma}\varrho_{\lambda}(s'',ds')\Big) \,\sup_{s'\in A} U^{(\lambda)}_{h'}\big(s',f\big)+U^{(\lambda)}_{h'} \big(s,f\big)\\
& \leq C\sup_{s'\in A} U^{(\lambda)}_{h'}\big(s',f\big)+U^{(\lambda)}_{h'} \big(s,f\big).
\end{align*}
The second inequality is by~(\ref{Carp}).

\end{proof}

\section{The Freidlin-Wentzell dynamics}\label{SecFreidlin}

We will now define a homogenized dynamics in which there is no deterministic evolution between jumps in phase space due to collisions.     The homogenized dynamics behaves similarly to the original dynamics, except that its state modulated resolvent is more analytically tractable.  For the original dynamics, between collisions, the particle follows an orbit over a connected segment of a level curve  of the Hamiltonian $H=\frac{1}{2}p^{2}+V(x)$.  If the particle starts at $(x,p)\in \Sigma$ with $|p|\gg 1$, then the particle will likely pass through over that curve on the order of $\textup{min}\big(|p|, \lambda^{-1}\big)$ times before the next collision, since the escape rates satisfy $\mathcal{E}_{\lambda}(p)\leq C(1+\lambda|p|)$ for  some $C>0$ and all $\lambda<1$ and $p\in \R$.  This is suggestive of  a Freidlin-Wentzell limit~\cite{Wentzell}, in which a Markovian dynamics emerges on the set of connected level curves  of a Hamiltonian for a system in which the dynamics is driven by a Hamiltonian evolution perturbed by a comparatively slow-acting noise (which they take to be a white noise).  Since $V(x)$ is bounded,  the level curves of $H(x,p)$ are almost flat  when $|p|\gg 1$ (and thus essentially like those of $H(x,p)=\frac{1}{2}p^{2}$ for large energies).  We do not discuss Freidlin-Wentzell limits further, and we proceed with defining the formalism relevant for us.

First, we define a state space $\Gamma_{V}$ identified with the set of connected components of level curves of $H(x,p)$  determined by the potential $V:\mathbb{T}\rightarrow \R^{+}$.  We define  $\Gamma_{V} $ as the image of a map $\Sigma\rightarrow \R^{+}\times \mathbb{Z}$ given by $\mathbf{G}_{V}(s)=\big(2^{\frac{1}{2}}H^{\frac{1}{2}}(s),\,n(s)\big)$ in which the component $n(s)\in \mathbb{Z}$ is a labeling of the connected components of the level curves corresponding to the energy $H(s)$.  When the particle has energy $H(s) > \sup_{x}V(x)$, then the Hamiltonian evolution drives the particle to revolve around the torus in one direction or another.  We make the convention that these level curves  are labeled with $\pm 1$ depending on the sign of $p$, and the remainder of the labeling at lower energies is arbitrary. 

\begin{definition}\label{Definitions} \text{ }

\begin{enumerate}
\item    We place a measure on $\Gamma_{V}\subset \R^{+}\times \mathbb{Z}$ through the Lebesgue measure on the preimage in $\Sigma$ of the map $s\rightarrow \mathbf{G}_{V}(s)=\big(2^{\frac{1}{2}}H^{\frac{1}{2}}(s), n(s)\big)\in \Gamma_{V}$.  We refer to this measure by $d\gamma$ where the dummy variable $\gamma$ is identified as an element in $  \Gamma_{V}$.

\item  For  $\gamma\in \Gamma_{V}$, we define the probability measure $\eta_{\gamma}$ on $\Sigma$ as the normalization of Lebesgue measure over the preimage of $\mathbf{G}_{V}^{-1}(\gamma)$.   Also, we define the probability measure $\kappa_{\gamma}$ as
$$  \kappa_{\gamma}(ds)=   \frac{     \eta_{\gamma}(ds) \big( p^{2}+\big|\frac{dV}{dx}(x)\big|^{2}    \big)^{-\frac{1}{2}}        }{   \int_{\Sigma} \eta_{\gamma}(ds')   \big( |p'|^{2}+\big|\frac{dV}{dx}(x')\big|^{2}    \big)^{-\frac{1}{2}}      },   $$
where $s=(x,p)$ and $s'=(x',p')$.

\item For $f:\Sigma\rightarrow \R$, define the map $B(\Sigma)\rightarrow  B(\Gamma_{V})$ such that
$$\widehat{f}(\gamma)= \int_{\Sigma}\eta_{\gamma}(ds)\,f(s).$$ 

\item
We define the jump kernel $\widehat{\mathcal{J}}_{\lambda}:\Gamma_{V}^{2}\rightarrow \R^{+}$  as
\begin{align}\label{FreidWentRates}
        \widehat{\mathcal{J}}_{\lambda}(\gamma ,\gamma' )= \int_{\Sigma^{2}}\kappa_{\gamma}(dx\,dp) \eta_{\gamma' }(dx'\,dp')\,\delta_{0}(x-x')\mathcal{J}_{\lambda}(p,p') . 
 \end{align}
\end{enumerate}

\end{definition}

  We will sometimes use $ \widehat{f}$ to denote an arbitrary element of $ B(\Gamma_{V})$ without reference to a specific preimage $f$.   The kernel $ \widehat{\mathcal{J}}_{\lambda}(\gamma,\gamma'  )$ defines a Markov process $G_{t}\in \Gamma_{V}$ which has the same essential features described in List~\ref{RegimeList}. 
 For this comparison, the value $\mathbf{q}(\gamma)=\epsilon\rho 1_{\rho\geq l }  $ can be identified as the momentum of the element $\gamma=(\rho,\epsilon)\in \Gamma_{V}$.  
 For $\widehat{h},\widehat{f}\in B(\Gamma_{V})$ with $\widehat{h}$  non-negative and $\widehat{h}\neq 0$, we define the  kernel  $\overline{U}^{(\lambda)}_{\widehat{h}}\big(\gamma,\widehat{f}\big)$   as 
$$\overline{U}^{(\lambda)}_{\widehat{h}}\big(\gamma,\widehat{f}\big)=\mathbb{E}_{\gamma}^{(\lambda)}\Big[\int_{0}^{\infty}dt\, \widehat{f}(G_{t})\,e^{-\int_{0}^{t}dr\,\widehat{h}(G_{r} )}    \Big] .
$$
The analogous statements of Section~\ref{SecBasic} all hold for $\overline{U}^{(\lambda)}_{\widehat{h}}$.   With $h:\Sigma\rightarrow \R^{+}$ defined as in~(\ref{OurAitch}), we define $\widehat{h}:\Gamma_{V}\rightarrow \R^{+}$ which has the form $\widehat{h}(\rho,\epsilon)= \chi( \rho\leq \sqrt{2l}    )$.  In future, we will drop the subscript from $\overline{U}^{(\lambda)}_{\widehat{h}}$ and take the form of  $\widehat{h}$ as above.

\begin{remark}\label{Remark}\text{  }

\begin{enumerate}

\item We can recover Lebesgue measure from the $\eta_{\gamma}$'s as the integral $$ds=\int_{\Gamma_{V}}d\gamma\,\eta_{\gamma}(ds).$$   

\item The kernel
$\widehat{\mathcal{J}}_{\lambda}(\gamma ,\gamma' )$ can be written as 
$$\widehat{\mathcal{J}}_{\lambda}(\gamma ,\gamma' )= \int_{\Sigma}\kappa_{\gamma}(dx\,dp) \sum_{p'=\pm  \sqrt{ |\rho'|^{2}-2V(x)   }   } \chi\big( (x,p')\in \mathbf{G}_{V}^{-1}(\gamma')  \big) \mathcal{J}_{\lambda}
(p, p' )\frac{  \big(  |p'|^{2}+\big|\frac{dV}{dx}(x)\big|^{2}    \big)^{\frac{1}{2}}     }{   |p'|  } .$$

\item Let $\mathit{I}_{\gamma}\subseteq \mathbb{T}$ be the range of the torus component of the set $\mathbf{G}_{V}^{-1}(\gamma)$.  For $x'\in  \mathit{I}_{\gamma}$,
$$   \int_{\Sigma}\kappa_{\gamma}(dx\,dp) \delta_{0}(x-x')=\frac{ \big(\rho^{2}-2V(x')\big)^{-\frac{1}{2} }   }{ \int_{\mathit{I}_{\gamma}}dx\, \big(\rho^{2}-2V(x)\big)^{-\frac{1}{2} }   }.   $$
When $\rho> \sup_{x}V(x)$, then $\mathit{I}_{\gamma}=\mathbb{T}$.

\end{enumerate}

\end{remark}

 For facility, we list some notation below.   
\begin{eqnarray*}
&\gamma=(\rho,\epsilon)      & \text{State in $\Gamma_{V}\subset \R^{+}\times \Z$.     }\\
&\mathbf{G}_{V}(s) =\gamma(s)    & \text{State in $ \Gamma_{V}$ associated with the element $s\in \Sigma$.  }\\
&\mathbf{q}(\gamma)   & \text{The quasi-momentum: $\mathbf{q}(\rho,\epsilon)=\rho\,\epsilon\,1_{\rho\geq l}$.  }\\
 & \mathbf{g}_{n}=(\mathbf{r}_{n},\mathbf{e}_{n}) &   \text{Skeleton chain for the Freidlin-Wentzell process.} \\ 
&  \widehat{\mathcal{J}}_{\lambda}(\gamma,\gamma'  )    &   \text{Jump kernel for the Freidlin-Wentzell process.} \\
&\widehat{\mathcal{E}}_{\lambda}(\gamma) &   \text{Escape rates for the Freidlin-Wentzell process.} \\
&  \widehat{T}_{\lambda}(\gamma,\gamma')  &   \text{Transition kernel for the skeleton chain.} \\
& U^{(\lambda)}f    &   \text{The $h$-modulated resolvent   of $f\in B(\Sigma)$ for the original process.}\\ 
& \overline{U}^{(\lambda)}\widehat{f} &  \text{The $\widehat{h}$-modulated resolvent of $\widehat{f}\in B(\Gamma_{V})$ for the Freidlin-Wentzell process}.
\end{eqnarray*}
The skeleton chain for the Freidlin-Wentzell process is the sequence of states at collision times and has transition kernel $ \widehat{T}_{\lambda}(\gamma,\gamma') =\frac{ \widehat{\mathcal{J}}_{\lambda}(\gamma,\gamma'  )  }{\widehat{\mathcal{E}}_{\lambda}(\gamma)    }$.  

Proposition~\ref{NewStuff} lists some characteristics of the jump rates for the original process and Freidlin-Wentzell process that we will use.  The proof uses elementary techniques, and we do not include it.  Parts (1) and (2) of Proposition~\ref{NewStuff} give bounds for the rate of collisions, and Parts (3) and (4) derive from the contractive nature of the jump rates when starting at high momentum.

\begin{proposition}\label{NewStuff} 
There are $c,C>0$ such that the following hold:
\begin{enumerate}
\item For all  $p\in \R$,
$$   c \,\textup{max}(1,\lambda|p|) \leq  \mathcal{E}_{\lambda}(p)\leq  C (1+\lambda |p|)  .  $$ 

\item For all $(\rho,\epsilon)\in \Gamma_{V}$,  $$c\,\textup{max}(1,\lambda\rho)\leq \widehat{\mathcal{E}}_{\lambda}(\rho,\epsilon)\leq C (1+\lambda \rho).$$

\item  Let $W:\Sigma\rightarrow [0,1]$ be defined as $W(s)= \frac{H^{\frac{1}{2}}(s) }{1+H^{\frac{1}{2}}(s) }$.   For all $\lambda<1$ and $s=(x,p)$ with $|p|>\lambda^{-1}$,
$$ \int_{\Sigma}ds'  \mathcal{J}_{\lambda}(s,s')\Big(W(s)-W(s')     \Big)\geq c\lambda.     $$

\item  Let  $W_{\lambda}:\Gamma_{V}\rightarrow \R^{+}$  be defined as
$ W_{\lambda}(\rho,\epsilon)= \lambda^{-1} \log\big(1+\lambda\rho    \big)     $.
For all $\lambda<1$ and $\gamma=(\rho,\epsilon)$ with $\rho>\lambda^{-1}$,
$$
  \int_{\Gamma_{V}}d\gamma'\,\widehat{T}_{\lambda}(\gamma,\gamma')\Big( W_{\lambda}(\gamma)-W_{\lambda}(\gamma')    \Big)    \geq c . $$

\end{enumerate}

\end{proposition}

In the lemma below, we set $l=1+2\sup_{x}V(x)$ as in the definition of $h$ (\ref{OurAitch}). 

\begin{lemma}\label{PreFreidlin}
Let $\mathbf{g}_{n}=(\mathbf{r}_{n},\mathbf{e}_{n})   \in \Gamma_{V}$ be the skeleton chain for the Freidlin-Wentzell process starting from $\gamma=(\rho,\epsilon)$.  Also let $\tilde{N}$ be the hitting time that $\mathbf{r}_{n}$ jumps below $\rho-1$.

\begin{enumerate}

\item  There is a $C>0$ such that for all $\gamma$ with $\rho>\sqrt{2l}$ and all non-negative $\widehat{f}\in B(\Gamma_{V})$,
$$
\mathbb{E}^{(\lambda)}_{\gamma}\Big[ \sum_{n=1}^{\tilde{N}-1 } \widehat{f}(\mathbf{g}_{n})    \Big]  \leq  C \Big( \sup_{ \rho'>\lambda^{-1} } \widehat{f}(\gamma ')   +\int_{\rho\leq \rho'\leq \lambda^{-1} } d\gamma' \widehat{f}(\gamma ') \Big).
$$

\item  Define the density $F^{(\lambda)}_{ \gamma }(\gamma' ):=\mathbb{E}^{(\lambda)}_{\gamma}\big[ \delta\big( \mathbf{g}_{\tilde{N}}-\gamma'    \big)  \big]$.  There is a  $C>0$ such that for all $\lambda\leq 1$ and  $\gamma$ with $\sqrt{2l}\leq \rho \leq \lambda^{-1}$,
$$
F^{(\lambda)}_{ \gamma }(\gamma' ) \leq  C e^{-\frac{1}{16}|\rho-\rho'  | }\chi(\rho'\leq \rho-1), $$
where $\gamma'=(\rho',\epsilon)$.

\item  For $\gamma$ with $\rho>\lambda^{-1}$, let $\mathbf{N}$ be the hitting time that $\mathbf{r}_{n}$ jumps below $\lambda^{-1}$.   There is a $C>0$ such that for all $\gamma$, $\lambda<1$, and non-negative $\widehat{f}\in B(\Gamma_{V})$, 
$$\mathbb{E}^{(\lambda)}_{\gamma}\Big[  \sum_{n=0}^{\mathbf{N}-1 } \frac{ \widehat{f}(\mathbf{g}_{n}) }{ \widehat{\mathcal{E}}_{\lambda}(\mathbf{g}_{n})     }\Big] \leq  C\lambda^{-1}\log\big(1+\lambda\rho    \big)\,\sup_{\rho>\lambda^{-1}}   \frac{ \widehat{f}(\gamma) }{ \widehat{\mathcal{E}}_{\lambda}(\gamma)     } .    $$

\item  Pick $L>0$.  There is a $C_{L}>0$ such that for all $\lambda<1$ and $\widehat{f}\in B(\Gamma_{V})$,  
$$ \int_{\rho'\leq  \sqrt{2L} }d\gamma'\,\overline{U}^{(\lambda)} \big(\gamma',\widehat{f}\big)   \leq C_{L} \int_{\Gamma_V}d\gamma'\,e^{-\frac{\lambda}{2}(\rho')^{2}}  \widehat{f}(\gamma')  .$$

\end{enumerate}

\end{lemma}

\begin{proof}\text{ }\\
\noindent Part (1): Define the measure $\mu_{\gamma}^{(\lambda)}$ on $\Gamma_{V}$ such that for $\widehat{f}\in B(\Gamma_{V})$
\begin{eqnarray} \nonumber
\mathbb{E}^{(\lambda)}_{\gamma}\Big[ \sum_{n=1}^{\tilde{N}-1 }  \widehat{f}(\mathbf{g}_{n})     \Big] &= & 
   \sum_{n=1}^{\infty}\int_{ \rho-1\leq \rho_{m}  }d\gamma_{1}\cdots d\gamma_{n}   \,  \widehat{T}_{\lambda}( \gamma,\gamma_{1})  \prod_{m=1}^{n-1} \widehat{T}_{\lambda}(\gamma_{m},\gamma_{m+1}) \widehat{f}(\gamma_{n}) \\ &:=  &\int_{\Gamma_{V}} d\mu_{\gamma}^{(\lambda)}(\gamma') \widehat{f}(\gamma '),
\end{eqnarray}
where $\gamma_{m}=(\rho_{m},\epsilon_{m})$.  The measure $\mu_{\gamma}^{(\lambda)}$ has its support on the set of $(\rho',\epsilon')\in \Gamma_{V}$ with $\rho'\geq \rho-1$.
 \begin{align}\label{Lab}
 \int_{\Gamma_{V}} d\mu_{\gamma}^{(\lambda)}(\gamma')&\widehat{f}(\gamma ')\leq \sum_{n=1}^{\infty}\int_{ \rho-1\leq \mathbf{q}_{m}  }d\gamma_{n}\cdots d\gamma_{1}   \, e^{\frac{\lambda}{4}( -\rho_{n}^2  +\rho^{2}   )} \, \widehat{T}_{0}( \gamma,\gamma_{1})  \prod_{m=1}^{n-1} \widehat{T}_{0}(\gamma_{m},\gamma_{m+1} )  \widehat{f}(\gamma_{n})\nonumber  \\ & = \int_{\Gamma_{V}} d\mu_{\gamma}^{(0)}(\gamma')  e^{\frac{\lambda}{4}( -(\rho')^2  +\rho^{2}   )}\widehat{f}(\gamma ')\nonumber  \\  &\leq  \big(\sup_{ \rho'>\lambda^{-1} }  \widehat{f}(\gamma ') \big)\int_{\rho'>\lambda^{-1} }d\mu_{\gamma}^{(0)}(\gamma')\,e^{\frac{\lambda}{4}( -(\rho')^2  +\rho^{2}   )}   + e^{\frac{1}{2}}\int_{\rho'\leq \lambda^{-1} } d\mu_{\gamma}^{(0)}(\gamma') \widehat{f}(\gamma ')
 \end{align}
The first inequality in~(\ref{Lab}) is from the detailed balance-type inequality
\begin{align}\label{Yausers}
\widehat{T}_{\lambda}(\gamma,\gamma' )\leq e^{\frac{\lambda}{4}( -(\rho')^2  +\rho^{2}   )} \widehat{T}_{0}(\gamma,\gamma'),   
\end{align}
which we apply for each instance of $\widehat{T}_{\lambda}$.  Equation~(\ref{Yausers}) follows  by the formula defining  the jump rates $\widehat{\mathcal{J}}_{\lambda}(\gamma,\gamma') $ and the following three facts:  $\widehat{T}_{\lambda}( \gamma,\gamma')=\frac{ \widehat{\mathcal{J}}_{\lambda}(\gamma,\gamma')   }{ \widehat{\mathcal{E}}_{\lambda}(\gamma)    }$ by definition of $\widehat{T}_{\lambda}$,  $\widehat{\mathcal{E}}_{\lambda}(\gamma)\geq (1+\lambda)\widehat{\mathcal{E}}_{0}(\gamma)=(1+\lambda)8$, and
\begin{align*}
\mathcal{J}_{\lambda}(p,p')=(1+\lambda)\big|p-p'\big|e^{\frac{\lambda}{4}( -(p')^2  +p^{2}   )} e^{-\lambda^{2}\frac{1}{8}(p+p')^{2}-\frac{1}{8}(p-p')^{2}}\leq (1+\lambda)e^{\frac{\lambda}{4}( -(p')^2  +p^{2}   )} \mathcal{J}_{0}(p,p').
\end{align*}
 The second inequality in~(\ref{Lab}) is Holder's for the domain $\rho'> \lambda^{-1}$, and for the domain  $\rho''\leq \lambda^{-1}$ we use that
$$ e^{\frac{\lambda}{4}( -(\rho')^2  +\rho^{2}   )}\leq e^{\frac{1}{2}},  $$
since $\mu_{\gamma}^{(0)}(\gamma') $ has support over $\gamma'$ with $  \rho'\geq \rho-1 $ and $\rho\leq \lambda^{-1}$.    

However, we claim that there is a $c>0$ such that for all $\gamma$ with $\rho'>\sqrt{2l}$, then
\begin{align}\label{JebDurango}
  \mu_{\gamma}^{(0)}(d\gamma')\leq  c\chi(\rho'\geq \rho-1)   d\gamma' .
\end{align}
Let us assume this now and return to it at the end of the proof.  Plugging~(\ref{JebDurango}) into~(\ref{Lab}) gives the first inequality below
\begin{align}
 \mathbb{E}^{(\lambda)}_{\gamma}\Big[ \sum_{n=1}^{\tilde{N}-1 } \widehat{f}(\mathbf{g}_{n})    \Big] & \leq  c\big(\sup_{ \rho'>\lambda^{-1} } \widehat{f}(\gamma ') \big)\int_{\rho'>\lambda^{-1} }d\gamma'\,e^{\frac{\lambda}{4}( -(\rho')^2  +\rho^{2}   )}   + ce^{\frac{1}{2}}\int_{\rho\leq \rho'\leq \lambda^{-1} } d\gamma' \widehat{f}(\gamma ') \nonumber  \\ & \leq  8c\sup_{ \rho'>\lambda^{-1} } \widehat{f}(\gamma ')   + ce^{\frac{1}{2}}\int_{\rho\leq \rho'\leq \lambda^{-1} } d\gamma' \widehat{f}(\gamma ').
   \end{align}
The second inequality is from
\begin{align*}
\int_{\rho'>\lambda^{-1} }d\gamma'\,e^{\frac{\lambda}{4}( -(\rho')^2  +\rho^{2}   )}   \leq 4e^{\frac{1}{4\lambda}} \int_{ \rho'>\lambda^{-1}  }d\rho'e^{-\frac{\lambda}{4} (\rho')^2 }\leq 4e^{\frac{1}{4\lambda}}\lambda\int_{ \rho'>\lambda^{-1}  }d\rho'\,\rho'\,e^{-\frac{\lambda}{4} (\rho')^2 }\leq 8.
\end{align*}
 In the first inequality, we have bounded $d\gamma'\leq  4 d\rho'$, since the measure $d\gamma'$ is close to  $ d\rho'$ for regions with $\rho'\gg 1$.  The factor of four is used because there are two branches corresponding to positive and negative momentum, and we have multiplied by an extra factor of two to cover the error with the dominant term.  

Now we show~(\ref{JebDurango}).  The measure $\mu_{\gamma}^{(0)}$ can be written as 
\begin{align}\label{Garga}
\mu_{\gamma}^{(0)}(d\gamma')= \chi\big( \rho'\geq \rho -1  \big)    \widehat{T}_{0}(\gamma,\gamma' )+\int_{ \rho''\geq \rho-1   }d\gamma''\, \widehat{T}_{0}( \gamma,\gamma'')\mathbb{E}^{(\lambda)}_{\gamma''}\Big[ \sum_{n=1}^{\tilde{N}-1 } \delta(\mathbf{g}_{n}-\gamma')    \Big].
\end{align}
 Define the density $w_{\gamma}$ on $\Gamma_{V}$ as
$$w_{\gamma}(\gamma')=\int_{\rho''\leq \rho -1 }d\gamma''\widehat{T}_{0}(\gamma'',\gamma' ) \chi\big( \rho'\geq \rho-1   \big) .  $$ 
 The flat measure $d\gamma$   is invariant with respect to the transition rates 
  $\widehat{T}_{0}(\gamma ,\gamma')$, and thus for $\gamma'=(\rho',\epsilon')\in \Gamma_{V}$,
 \begin{multline} \label{PerroBlanco}
 \chi\big(\rho'\geq \rho-1\big)= w_{\gamma}(\gamma')+\chi\big( \rho'\geq \rho -1  \big)   \int_{ \rho''\geq \rho-1   }d\gamma''\,w_{\gamma}(\gamma'')\, \widehat{T}_{0}(\gamma'',\gamma' ) \\  +\int_{ \rho''\geq \rho-1   }d\gamma''\,w_{\gamma}(\gamma'')\,\int_{ \rho'''\geq \rho-1   }d\gamma'''\widehat{T}_{0}(\gamma'',\gamma''' ) \mathbb{E}^{(\lambda)}_{\gamma'''}\Big[ \sum_{n=1}^{\tilde{N}-1 } \delta(\mathbf{g}_{n}-\gamma')    \Big].  
 \end{multline}
This formula treats the influx of mass jumping from  the set $\{\gamma'\,\big|\rho'\leq \rho-1 \}$ as a source, and sums the expected occupation density  before the mass leaves the set $\{\gamma'\,\big|\rho'> \rho-1 \}$.   However, we can find a $c$ such that
\begin{align}\label{Nazarith}
\widehat{T}_{0}(\gamma,\gamma')\leq c\int_{ \rho''\geq \rho-1   }d\gamma''\,w_{\gamma}(\gamma'')\, \widehat{T}_{0}( \gamma'',\gamma')   
\end{align}
for all $\gamma$ with $\rho> \sqrt{2l}$ all all $\gamma'$ with $\rho'>\rho-1$.  If~(\ref{Nazarith}) holds, then plugging in to~(\ref{PerroBlanco}) and throwing away the first  term on the right side gives 
$$ \chi\big(\rho'\geq \rho-1\big)\geq \frac{1}{c}\chi\big( \rho'\geq \rho -1  \big)   \widehat{T}_{0}(\gamma,\gamma')+ \frac{1}{c}\int_{ \rho''\geq \rho-1   }d\gamma''\widehat{T}_{0}( \gamma,\gamma'')\, \mathbb{E}^{(\lambda)}_{\gamma''}\Big[ \sum_{n=1}^{\tilde{N}-1 } \delta(\mathbf{g}_{n}-\gamma')    \Big].  $$
We can employ this inequality in~(\ref{Garga}) to reach~(\ref{JebDurango}).

 To see~(\ref{Nazarith}), first observe that the transition kernel $\widehat{T}_{0}\big(\gamma,\gamma' \big)$ has the simpler form
\begin{align*} 
  \widehat{T}_{0}(\gamma ,\gamma' ) &= \frac{1}{8}\int_{\Sigma^{2}}\kappa_{\gamma}(dx\,dp)\eta_{\gamma' }(dx'\,dp')\,\delta_{0}(x-x')\mathcal{J}_{0}(p,p')\\  &= \frac{1}{8}\frac{ \int_{\mathbb{T} }dx\,\frac{  \big( |\rho'|^{2}-2V(x)+\big|\frac{dV}{dx}(x)\big|^{2}    \big)^{\frac{1}{2}}     }{  \big(\rho^{2}-2V(x)\big)^{\frac{1}{2}} \big(|\rho'|^{2}-2V(x)\big)^{\frac{1}{2}} }\mathcal{J}_{0}\Big( \epsilon (\rho^{2}-2V(x))^{\frac{1}{2}},\, \epsilon'(|\rho'|^{2}-2V(x))^{\frac{1}{2}} \Big)  }{ \int_{\mathbb{T} }dx\, \big(\rho^{2}-2V(x)\big)^{-\frac{1}{2}  }   }   ,   
  \end{align*}
 since the escape rates $\widehat{\mathcal{E}}_{0}(\gamma)=8$ are constant.  The second equality only holds when $\rho'\geq  \sup_{x}V(x)$, and otherwise there are two terms.  For $\gamma=(\rho,\epsilon)$ with $\rho\geq \sqrt{2l}$, the label component is $\epsilon=\pm 1$, and we can identify $\gamma$ with the quasi-momentum value $\mathbf{q}(\gamma)=\epsilon \rho$.  The rates describe what is nearly an unbiased random walk for the quasi-momentum.

The function $ w_{\gamma}(\gamma'')$ is uniformly bounded away from zero over any finite region of $
\gamma''$ with $\rho-1\leq \rho''\leq \rho+L $ for $L>0$.  It is sufficient to take, say, $L=1$.  For large enough $c'>0$, we thus have the first inequality below
    \begin{align*} \int_{ \rho''\geq   \rho-1 }d\gamma''\,w_{\gamma}(\gamma'')\, \widehat{T}_{0}(\gamma'',\gamma' ) & \geq  \frac{1}{c'} \int_{\rho-1\leq  \rho''\leq \rho+1    }d\gamma''\, \widehat{T}_{0}(\gamma'',\gamma' )\\ & \geq \frac{1}{c} \widehat{T}_{0}(\gamma,\gamma') .
\end{align*}
Finally, we can choose $c>0$ large enough to make the second inequality hold for all $\gamma,\gamma'$ with $\sqrt{2l}<\rho$ and $\rho'\leq \rho-1$.

\vspace{.5cm}

\noindent Part (2): We have the closed formula
\begin{align}\label{Sin}
\mathbb{E}^{(\lambda)}_{\gamma}\big[ \delta\big( \mathbf{g}_{\tilde{N}}-\gamma'    \big)  \Big] = \mathbb{E}^{(\lambda)}_{\gamma}\Big[\sum_{n=0}^{\tilde{N}-1}\widehat{T}_{\lambda}( \mathbf{g}_{n},\gamma')   \Big].  
\end{align}
This follows formally by the optional stopping theorem with stopping time $\tilde{N}$ and ``martingale"
  $$\sum_{n=1}^{m} \delta\big( \mathbf{g}_{n}-\gamma'    \big) -\widehat{T}_{\lambda}( \mathbf{g}_{n-1},\gamma'). $$
To be more rigorous, we should replace $\delta\big( \cdot-\gamma'    \big)$ by a family of indicators approximating it.

With~(\ref{Sin}), we can apply Part (1) with $\widehat{f}(\gamma)= \widehat{T}_{\lambda}(\gamma,\gamma')$ to get the inequality below for some $C'>0$.    
\begin{align}\label{CarCrash}
\mathbb{E}^{(\lambda)}_{\gamma}\big[\sum_{n=0}^{\tilde{N}-1}\widehat{T}_{\lambda}( \mathbf{g}_{n},\gamma')  \Big] & =\widehat{T}_{\lambda}( \gamma,\gamma')  + \mathbb{E}^{(\lambda)}_{\gamma}\Big[\sum_{n=1}^{\tilde{N}-1}\widehat{T}_{\lambda}( \mathbf{g}_{n},\gamma')   \Big]\nonumber  \\ & \leq  \widehat{T}_{\lambda}(\gamma,\gamma')  +C' \Big( \sup_{ \rho''>\lambda^{-1} }\widehat{T}_{\lambda}( \gamma'',\gamma')  +\int_{\rho\leq \rho''\leq \lambda^{-1} } d\gamma''\widehat{T}_{\lambda}( \gamma'',\gamma')  \Big).
\end{align}
However, there is a $c>0$ such that for all $\lambda< 1$ and $\gamma,\gamma'$ with $\rho,\rho'\leq \lambda^{-1}$, 
$$ \widehat{T}_{\lambda}(\gamma,\gamma')\leq ce^{-\frac{1}{16} |\rho-\rho'|   }\quad \text{and}\quad \sup_{\rho'\geq \lambda^{-1}   } \widehat{T}_{\lambda}(\gamma,\gamma')\leq ce^{-\frac{1}{16} |\rho-\lambda^{-1}|   } .$$
Plugging these in to~(\ref{CarCrash}) gives the uniform bound.

\vspace{.5cm}

\noindent Part (3):  We begin with the inequality,
$$   \mathbb{E}_{\gamma}^{(\lambda)}\Big[\sum_{n=1}^{\mathbf{N}} \frac{\widehat{f}(\mathbf{g}_{n})   }{\widehat{\mathcal{E}}_{\lambda}( \mathbf{g}_{n}  )} \Big]\leq  \Big(\sup_{\rho'>\lambda^{-1}}   \frac{\widehat{f}(\gamma')   }{\widehat{\mathcal{E}}_{\lambda}(\gamma' )} \Big)\mathbb{E}_{\gamma}^{(\lambda)}\big[\mathbf{N}\big] . $$
Let $W_{\lambda}:\Gamma_{V}\rightarrow \R^{+}$ be as in (4) of Proposition~\ref{NewStuff}.  It follows by (4) of Proposition~\ref{NewStuff}  that $cn+W_{\lambda}(\mathbf{g}_{n})$ is a supermartingale over the time interval $n\in[0,\mathbf{N}]$.  We have the inequalities,
$$ \mathbb{E}_{\gamma}^{(\lambda)}\big[\mathbf{N}\big]\leq \frac{1}{c}\mathbb{E}_{\gamma}^{(\lambda)}\big[W_{\lambda}(\gamma)-W_{\lambda}(\mathbf{g}_{\mathbf{N}})]\leq \frac{1}{c} W_{\lambda}(\gamma)=\frac{1}{c\lambda} \log\big(1+\lambda\rho    \big) , $$
where the first inequality is by the optional stopping theorem, and the second inequality is since $W_{\lambda}\geq 0$.

\vspace{.5cm}

\noindent Part (4):  This follows analogously to Part (2) of Lemma~\ref{LemPick}.

\end{proof}

 The inequality in  Part (2) of the lemma below is analogous to Theorem~\ref{ThmMain}.

\begin{lemma}\label{Freidlin}
Let  $\overline{U}^{(\lambda)}$ be the state-modulated resolvent of the function $\widehat{h}$.

\begin{enumerate}
\item  $\overline{U}^{(\lambda)}\widehat{f}$ satisfies the integral equation
$$ \widehat{f}(\gamma)  = \widehat{h}(\gamma)\overline{U}^{(\lambda)} \big(\gamma,\widehat{f}\big)+
\int_{\Gamma_{V}}d\gamma'\, \widehat{\mathcal{J}}_{\lambda}(\gamma,\gamma' )\Big(\overline{U}^{(\lambda)} \big(\gamma,\widehat{f}\big)-\overline{U}^{(\lambda)} \big(\gamma',\widehat{f}\big) \Big). $$

\item
 There is a $c>0$ such that for all measurable $\widehat{f}:\Gamma_{V}\rightarrow \R^{+}$, $\lambda<1$, and $\gamma \in \Gamma_{V}$
$$\overline{U}^{(\lambda)} \big(\gamma,\widehat{f}\big)\leq c\Big(\sup_{\gamma'\in \Gamma_{V}} A^{(\lambda)}(\gamma,\gamma')\,\widehat{f}(\gamma')+\int_{\Gamma_{V}}d\gamma'\, B^{(\lambda)}(\gamma,\gamma')\,\widehat{f}(\gamma')\Big),  $$
 where $ A^{(\lambda)}(\gamma,\gamma')$ and $  B^{(\lambda)}(\gamma,\gamma')$ are defined as
\begin{align*}
A^{(\lambda)}(\rho,\epsilon,\,\rho',\epsilon')&= \Big(1+\textup{min}\big(\rho,\lambda^{-1}\log(1+\lambda \rho)      \big)\,\chi\big( \rho'\geq \lambda^{-1}\big)\Big)\frac{1}{\widehat{\mathcal{E}}_{\lambda}(\rho',\epsilon') },\\    
B^{(\lambda)}(\rho,\epsilon,\,\rho',\epsilon')&=\big( 1+\textup{min}(\rho,\rho') \big)\,\chi( \rho\leq \lambda^{-1})\frac{1}{\widehat{\mathcal{E}}_{\lambda}(\rho',\epsilon') }.
 \end{align*}

\end{enumerate}

\end{lemma}

\begin{proof} Part (1) follows easily from the definition of $\overline{U}^{(\lambda)}\widehat{f} $, so we focus Part (2).   By rearranging the integral equation from Part (1), we have the equation
\begin{eqnarray}
\overline{U}^{(\lambda)} \big(\gamma,\widehat{f}\big)&=& \frac{ \widehat{f}(\gamma) }{  \widehat{h}(\gamma)+\widehat{\mathcal{E}}_{\lambda}(\gamma)     } +
\int_{\Gamma_{V}}d\gamma'\, \frac{\widehat{\mathcal{J}}_{\lambda}(\gamma,\gamma' )}{  \widehat{h}(\gamma)+\widehat{\mathcal{E}}_{\lambda}(\gamma)         } \overline{U}^{(\lambda)} \big(\gamma',\widehat{f}\big)\nonumber  \\ \label{Hosers}
&=& \widehat{C}_{\lambda}(\gamma  )\frac{ \widehat{f}(\gamma) }{  \widehat{\mathcal{E}}_{\lambda}(\gamma)     } +
\widehat{C}_{\lambda}(\gamma  )\int_{\Gamma_{V}}d\gamma'\, \widehat{T}_{\lambda}(\gamma,\gamma' ) \,\overline{U}^{(\lambda)} \big(\gamma',\widehat{f}\big),     
\end{eqnarray}
where  $\widehat{T}_{\lambda}$ and $\widehat{C}_{\lambda}$ are defined as
$$\widehat{T}_{\lambda}\big( \gamma,\gamma'\big)=  \frac{\widehat{\mathcal{J}}_{\lambda}( \gamma,\gamma')}{ \widehat{\mathcal{E}}_{\lambda}(\gamma)}\hspace{1cm} \text{and} \hspace{1cm} \widehat{C}_{\lambda}(\gamma  )= \frac{\widehat{\mathcal{E}}_{\lambda}(\gamma)   }{ \widehat{h}(\gamma)+\widehat{\mathcal{E}}_{\lambda}(\gamma)    } .  $$
Consider the chain $\mathbf{g}_{n}=(\mathbf{r}_{n},\mathbf{e}_{n})   \in \Gamma_{V}$ starting at $\gamma$ and making jumps with transition kernel $\widehat{T}_{\lambda}$.  The kernel for $\overline{U}^{(\lambda)} $ can be written as  
\begin{align}\label{Chomsky}
 \overline{U}^{(\lambda)} \big(\gamma,\widehat{f}\big)=
\sum_{n=0}^{\infty}\mathbb{E}^{(\lambda)}_{\gamma}\Big[ \Big(\prod_{r=0}^{n}  \widehat{C}_{\lambda}(\mathbf{g}_{r}) \Big)  \frac{ \widehat{f}(\mathbf{g}_{n}) }{ \widehat{\mathcal{E}}_{\lambda}(\mathbf{g}_{n})     }\Big].     
\end{align}

First, we will show that the bound for $ \overline{U}^{(\lambda)} \big((\rho,\epsilon),\widehat{f}\big)$ when $\rho>\lambda^{-1}$ follows from the bound for $\overline{U}^{(\lambda)} \big((\rho,\epsilon),\widehat{f}\big)$ when  $\rho\leq \lambda^{-1}$.   For $\gamma=(\rho,\epsilon)$ with $\rho>\lambda^{-1}$, let $\mathbf{N}\in \mathbb{N}$ be the hitting time that $\mathbf{r}_{n}$ jumps below $\lambda^{-1}$.    The form~(\ref{Chomsky}) allows us to write
\begin{align} 
 \overline{U}^{(\lambda)} \big(\gamma,\widehat{f}\big) &=  \mathbb{E}^{(\lambda)}_{\gamma}\Big[  \sum_{n=0}^{\mathbf{N}-1 }  \Big(\prod_{r=0}^{n}  \widehat{C}_{\lambda}(\mathbf{g}_{r}) \Big)\frac{ \widehat{f}(\mathbf{g}_{n}) }{ \widehat{\mathcal{E}}_{\lambda}(\mathbf{g}_{n})     }\Big]+\mathbb{E}^{(\lambda)}_{\gamma}\Big[  \Big(\prod_{r=0}^{\mathbf{N}}  \widehat{C}_{\lambda}(\mathbf{g}_{r}) \Big)\overline{U}^{(\lambda)} \big(\mathbf{g}_{\mathbf{N} },\widehat{f}\big)\Big]\nonumber \\ & \leq  \mathbb{E}^{(\lambda)}_{\gamma}\Big[  \sum_{n=0}^{\mathbf{N}-1 } \frac{ \widehat{f}(\mathbf{g}_{n}) }{ \widehat{\mathcal{E}}_{\lambda}(\mathbf{g}_{n})     }\Big]+\mathbb{E}^{(\lambda)}_{\gamma}\Big[ \overline{U}^{(\lambda)} \big(\mathbf{g}_{\mathbf{N} },\widehat{f}\big) \Big]\label{Wolfgang} \\ & \leq  C\lambda^{-1}\log\big(1+\lambda \rho\big)\,\Big(\sup_{\rho'>\lambda^{-1}}\frac{\widehat{f}(\gamma') }{ \widehat{\mathcal{E}}_{\lambda}(\gamma')     }  \Big) +\sup_{\rho'\leq \lambda^{-1}}  \overline{U}^{(\lambda)} \big(\gamma',\widehat{f}\big) \nonumber .
 \end{align}
The first inequality uses that $ \widehat{C}_{\lambda}(\gamma)\leq 1$, and  the second inequality uses Part (3) of Lemma~\ref{PreFreidlin} for the first term, and the definition of the hitting time $\mathbf{N}$ for the second.  Thus, it is sufficient for us to prove the statement of this lemma for the domain of $\gamma=(\rho,\epsilon)$  with  $\rho\leq \lambda^{-1}$.

Next, we focus on the domain $ \sqrt{2l}< \rho\leq \lambda^{-1}$.   For $(\mathbf{r}_{0},\mathbf{e}_{0})=(\rho,\epsilon)$ with $\rho>\sqrt{2l}$,   let $\tilde{N}_{n}$ be the sequence of hitting times   such that $\tilde{N}_{0}=0$ and 
$$ \hspace{2cm} \tilde{N}_{n}= \inf\{ m>  \tilde{N}_{n-1} \,\big|\,  \mathbf{r}_{m}\leq \mathbf{r}_{\tilde{N}_{n-1}}-1    \},\hspace{1cm} n\geq 1.    $$
In the above, we can take the infimum of the empty set to be $\infty$, and clearly there can be at most $\lceil\rho \rceil$ of the hitting times $\tilde{N}_{n}$ which are not infinite.  Also let   $\mathbf{T}\in \mathbb{N}$ be the the first time $\tilde{N}_{n}$ such that $\mathbf{r}_{n}\leq \sqrt{2l}$, and $I$ be the number of $\tilde{N}_{n}$ in $[1, \mathbf{T})$.   Analogously to~(\ref{Wolfgang}), we have the inequality
\begin{align}\label{ChomskyII}
 \overline{U}^{(\lambda)} \big(\gamma,\widehat{f}\big) \leq  \frac{\widehat{f}(\gamma)}{\widehat{\mathcal{E}}_{\lambda}(\gamma) } +   
\mathbb{E}^{(\lambda)}_{\gamma}\Big[ \sum_{n=1}^{\mathbf{T}-1 } \frac{\widehat{f}(\mathbf{g}_{n})}{\widehat{\mathcal{E}}_{\lambda}(\mathbf{g}_{n}) } \Big]  + \mathbb{E}^{(\lambda)}_{\gamma}\Big[ \overline{U}^{(\lambda)} \big(\mathbf{g}_{\mathbf{T} },\widehat{f}\big) \Big].
\end{align}
  By breaking the time step interval $[1,\mathbf{T}]\subset \mathbb{N}$ into $I$ subintervals $[ \tilde{N}_{m-1}+1,   \tilde{N}_{m}]$ and using nested conditional expectations and the strong Markov property,
\begin{align}\label{Hitchens}
\mathbb{E}^{(\lambda)}_{\gamma}\Big[ \sum_{n=1}^{\mathbf{T}-1 }\frac{\widehat{f}(\mathbf{g}_{n})}{\widehat{\mathcal{E}}_{\lambda}(\mathbf{g}_{n}) }\Big] & = \mathbb{E}^{(\lambda)}_{\gamma}\Big[ \sum_{m=1}^{I}  \mathbb{E}^{(\lambda)}\Big[ \sum_{n=\tilde{N}_{m-1}}^{\tilde{N}_{m}-1 } \frac{\widehat{f}(\mathbf{g}_{n})}{\widehat{\mathcal{E}}_{\lambda}(\mathbf{g}_{n}) }\,\Big|\, \mathcal{F}_{\mathbf{g}_{\tilde{N}_{m-1}}   }    \Big] \Big]\nonumber  \\ & = \mathbb{E}^{(\lambda)}_{\gamma}\Big[ \sum_{m=1}^{I-1} \frac{\widehat{f}(\mathbf{g}_{\tilde{N}_{m}})}{\widehat{\mathcal{E}}_{\lambda}(\mathbf{g}_{\tilde{N}_{m}}) } \Big]+\mathbb{E}^{(\lambda)}\Big[ \sum_{m=1}^{I}\mathbb{E}^{(\lambda)}_{\mathbf{g}_{\tilde{N}_{m-1}}}\Big[ \sum_{n=1}^{\tilde{N}_{1} -1}  \frac{\widehat{f}(\mathbf{g}_{n})}{\widehat{\mathcal{E}}_{\lambda}(\mathbf{g}_{n}) }   \Big]\Big] .
\end{align}
To bound $\overline{U}^{(\lambda)} \big(\gamma,\widehat{f}\big)$, we must bound the terms
  $$ (\text{i}). \hspace{.2cm} \mathbb{E}^{(\lambda)}\Big[ \sum_{m=1}^{I}\mathbb{E}^{(\lambda)}_{\mathbf{g}_{\tilde{N}_{m-1}}}\Big[ \sum_{n=1}^{\tilde{N}_{1} -1} \frac{\widehat{f}(\mathbf{g}_{n})}{\widehat{\mathcal{E}}_{\lambda}(\mathbf{g}_{n}) }   \Big]\Big],\quad (\text{ii}). \hspace{.2cm} \mathbb{E}^{(\lambda)}_{\gamma}\Big[ \sum_{m=1}^{I-1}  \frac{\widehat{f}(\mathbf{g}_{\tilde{N}_{m}})}{\widehat{\mathcal{E}}_{\lambda}(\mathbf{g}_{\tilde{N}_{m}}) } \Big], \quad  (\text{iii}).\hspace{.2cm} \mathbb{E}^{(\lambda)}_{\gamma}\Big[ \overline{U}^{(\lambda)} \big(\mathbf{g}_{\mathbf{T} },\widehat{f}\big)   \Big].   $$

By Part (1) of Lemma~\ref{PreFreidlin}, there is $C>0$ such that (i) is smaller than
\begin{align}\label{HitchensII}
 \mathbb{E}^{(\lambda)}_{\gamma}\Big[ \sum_{m=1}^{I}\mathbb{E}^{(\lambda)}_{\mathbf{g}_{\tilde{N}_{m-1}}}\Big[  \sum_{n=1}^{\tilde{N}_{1} -1} & \frac{\widehat{f}(\mathbf{g}_{n})}{\widehat{\mathcal{E}}_{\lambda}(\mathbf{g}_{n}) }   \Big]\Big] \nonumber \\ & \leq C\mathbb{E}^{(\lambda)}_{\gamma}\big[ I\big]\sup_{ \rho'>\lambda^{-1} }\frac{\widehat{f}(\gamma ')}{\widehat{\mathcal{E}}_{\lambda}(\gamma ') }  +C\mathbb{E}^{(\lambda)}_{\gamma}\Big[ \sum_{m=1}^{I}\int_{\mathbf{r}_{\tilde{N}_{m-1}}\leq \rho'\leq \lambda^{-1} } d\gamma'\frac{\widehat{f}(\gamma ')}{ \widehat{\mathcal{E}}_{\lambda}(\gamma ')     } \Big] \nonumber  \\ & \leq C\rho \sup_{ \rho'>\lambda^{-1} } \frac{\widehat{f}(\gamma ')}{\widehat{\mathcal{E}}_{\lambda}(\gamma ') }+ C \int_{ \rho'\leq \lambda^{-1} } d\gamma'\big(1+\textup{min}(\rho', \rho)\big) \frac{\widehat{f}(\gamma ')}{ \widehat{\mathcal{E}}_{\lambda}(\gamma ')     } .
 \end{align}
 For both terms in the second inequality, we have used that the sequence $\mathbf{r}_{\tilde{N}_{m}}$ decreases by increments $\geq 1$ for $m=1,\dots, I$.  Thus $I\leq \rho$,   and the number of $m$ such that $\mathbf{r}_{\tilde{N}_{m}}-1$ is smaller than some value $\rho'\leq \lambda^{-1}$ is less than $1+ \textup{min}(\rho, \rho')$.

For (ii), we can write
\begin{align}\label{HitchensIII}
\mathbb{E}^{(\lambda)}_{\gamma}\Big[ \sum_{m=1}^{I-1}  \frac{\widehat{f}(\mathbf{g}_{\tilde{N}_{m}})}{\widehat{\mathcal{E}}_{\lambda}(\mathbf{g}_{\tilde{N}_{m}}) } \Big]=\int_{\Gamma_{V}}\upsilon^{(\lambda)}_{\gamma}(d\gamma') \frac{\widehat{f}(\gamma ')}{ \widehat{\mathcal{E}}_{\lambda}(\gamma ')     }     ,
 \end{align} 
where $\upsilon^{(\lambda)}_{\gamma}(d\gamma')= \mathbb{E}^{(\lambda)}_{\gamma}\big[ \sum_{m=1}^{I-1} \delta(\mathbf{g}_{\tilde{N}_{m}}-\gamma')   \big]$.   By nested conditional expectations and the strong Markov property, we have the equalities below
 \begin{align}\label{SilverSpoons}
  \upsilon^{(\lambda)}_{\gamma}(d\gamma')&=\mathbb{E}^{(\lambda)}_{\gamma}\Big[ \sum_{m=0}^{I-2} \mathbb{E}^{(\lambda)}\big[\delta\big(\mathbf{g}_{\tilde{N}_{m+1} }-\gamma'\big)\,\big|\,\mathbf{g}_{\tilde{N}_{m}}   \big] \Big]=\mathbb{E}^{(\lambda)}_{\gamma}\Big[ \sum_{m=0}^{I-2} F^{(\lambda)}_{ \mathbf{g}_{\tilde{N}_{m}} }(\gamma')  \Big]\nonumber \\ & \leq C' \mathbb{E}^{(\lambda)}_{\gamma}\Big[ \sum_{m=0}^{I-2} e^{-\frac{1}{16} |\mathbf{r}_{\tilde{N}_{m} }-\rho' | }\chi\big(\rho'\leq \mathbf{r}_{\tilde{N}_{m} } -1 \big)    \Big]\nonumber \\ & < C'\chi(\rho'\leq \rho -1)\sum_{n=1}^{\infty} e^{-\frac{1}{16} n }< 16C'\chi\big(\rho'\leq \lambda^{-1} \big) , 
  \end{align}
where $F^{(\lambda)}_{ \gamma }(\gamma' )=\mathbb{E}^{(\lambda)}_{\gamma}\big[\delta\big(\mathbf{g}_{\tilde{N}}-\gamma'\big)   \big]$ is defined as in Part (2) of Proposition~\ref{PreFreidlin}, and the  first inequality is for some $C'>0$  by Part (2) of Proposition~\ref{PreFreidlin}.   The second inequality uses that $\mathbf{r}_{\tilde{N}_{m} }$ decreases by at least one at each time step.  With~(\ref{HitchensIII}) and~(\ref{SilverSpoons}),   
$$ \mathbb{E}^{(\lambda)}_{\gamma}\Big[ \sum_{m=1}^{I-1}  \frac{\widehat{f}(\mathbf{g}_{\tilde{N}_{m}})}{\widehat{\mathcal{E}}_{\lambda}(\mathbf{g}_{\tilde{N}_{m}}) } \Big]\leq C\int_{\rho'\leq \lambda^{-1}}d\gamma'\frac{\widehat{f}(\gamma ')}{ \widehat{\mathcal{E}}_{\lambda}(\gamma ')     }  $$
 for $C=16 C'$.

For (iii),  we have the following relations
\begin{align}\label{Hornswaggle}
\mathbb{E}^{(\lambda)}_{\gamma}\Big[ \overline{U}^{(\lambda)}\big(\mathbf{g}_{\mathbf{T} }, \widehat{f}\big)\Big]& = \int_{\Sigma}d\gamma'\overline{U}^{(\lambda)} \big(\gamma',\widehat{f}\big) \,\mathbb{E}^{(\lambda)}_{\gamma}\big[ \delta(\mathbf{g}_{\mathbf{T} }-\gamma')  \Big] \leq c'\int_{\rho'\leq  \sqrt{2l} }d\gamma'\,\overline{U}^{(\lambda)}\big(\gamma', \widehat{f}\big) \nonumber \\ &   \leq c'' \int_{\Gamma_V}d\gamma'\,e^{-\frac{\lambda}{2}(\rho')^{2}}  \widehat{f}(\gamma')\leq  c''' \int_{\Gamma_V}d\gamma'\,e^{-\frac{\lambda}{4}(\rho')^{2}}  \frac{ \widehat{f}(\gamma ')}{  \widehat{\mathcal{E}}_{\lambda}(\gamma ') }\nonumber  \\  & \leq c''' \int_{\rho'\leq \lambda^{-1}}\frac{\widehat{f}(\gamma ')}{ \widehat{\mathcal{E}}_{\lambda}(\gamma ')     }+c''''\sup_{ \rho'>\lambda^{-1} }\frac{ \widehat{f}(\gamma ')}{  \widehat{\mathcal{E}}_{\lambda}(\gamma ') }     .
\end{align}
For the first inequality, the density $d\gamma'=\mathbb{E}^{(\lambda)}_{\gamma}\big[ \delta(\mathbf{g}_{\mathbf{T} }-\gamma')  \big]$ is smaller than some $c'>0$ by Part (2) of Proposition~\ref{PreFreidlin}.  The second equality is  Part (4) of Proposition~\ref{PreFreidlin}, and the third by the bounds for $\widehat{\mathcal{E}}_{\lambda}$ from Part (2) of Proposition~\ref{NewStuff}.    
 For the last inequality, we have split the integration into the domains of $\gamma=(\rho,\epsilon)$ with  $\rho\leq \lambda^{-1}$ and  $\rho>\lambda^{-1}$  similarly to the proof of Part (2) of Lemma~\ref{LemPick}.

 With (i)-(iii), we have shown that there is a $C>0$ such that for all $\lambda<1$ and all $\gamma=(\rho,\epsilon)$ with $\rho>\sqrt{2l}$, 
\begin{align}\label{Nolbach}
\overline{U}^{(\lambda)} \big(\gamma,\widehat{f}\big) \leq  \big\|\frac{\widehat{f}}{\widehat{\mathcal{E}}_{\lambda} }\big\|_{\infty}+ C\rho\sup_{ \rho'>\lambda^{-1} }\frac{\widehat{f}(\gamma ')}{\widehat{\mathcal{E}}_{\lambda}(\gamma ') } + C \int_{ \rho'\leq \lambda^{-1} } d\gamma'\big(1+\textup{min}(\rho,\rho')\big) \frac{\widehat{f}(\gamma ')}{ \widehat{\mathcal{E}}_{\lambda}(\gamma ')     }. 
\end{align}

We can use~(\ref{Nolbach}) to extend our bound to the domain of $\gamma=(\rho,\epsilon)$ with $\rho\leq \sqrt{2l}$. Starting with the integral equation~(\ref{Hosers}), 
\begin{align}\label{Hungary}
&\overline{U}^{(\lambda)} \big(\gamma,\widehat{f}\big)= \widehat{C}_{\lambda}(\gamma  )\frac{ \widehat{f}(\gamma) }{  \widehat{\mathcal{E}}_{\lambda}(\gamma)     } +
\widehat{C}_{\lambda}(\gamma  )\int_{\Gamma_{V}}d\gamma'\, \widehat{T}_{\lambda}(\gamma,\gamma' ) \,\overline{U}^{(\lambda)} \big(\gamma',\widehat{f}\big)\nonumber  \\  &\leq   \big\|\frac{\widehat{f}}{\widehat{\mathcal{E}}_{\lambda}}\big\|_{\infty} +
\Big( \sup_{\substack{ \lambda<1 \\ \rho,\rho'\leq \sqrt{2l} } }\widehat{T}_{\lambda}(\gamma,\gamma' )  \Big)  \int_{\rho'\leq \sqrt{2l} }d\gamma'\,  \overline{U}^{(\lambda)} \big(\gamma',\widehat{f}\big)   + \int_{\rho'>\sqrt{2l}}d\gamma'\Big(\sup_{\substack{ \lambda<1 \\ \rho\leq \sqrt{2l} }  } \widehat{T}_{\lambda}(\gamma,\gamma' ) \Big)\overline{U}^{(\lambda)} \big(\gamma',\widehat{f}\big)\nonumber  \\ & \leq \big\|\frac{\widehat{f}}{\widehat{\mathcal{E}}_{\lambda}}\big\|_{\infty} +C'\Big(\sup_{\rho'> \lambda^{-1}}\frac{ \widehat{f}(\gamma ')}{  \widehat{\mathcal{E}}_{\lambda}(\gamma ') } +\int_{ \rho'\leq \lambda^{-1} } d\gamma'\frac{\widehat{f}(\gamma ')}{ \widehat{\mathcal{E}}_{\lambda}(\gamma ')     } \Big)
\end{align}
for large enough constant $C'>0$.  For the second term in the second inequality, the supremum of $\widehat{T}_{\lambda}(\gamma,\gamma' )$ over $\lambda<1$ and $\gamma,\gamma'\in \Gamma_{V}$ is bounded, and we bound the integral $ \int_{\rho'\leq \sqrt{2l} }d\gamma'\,  \overline{U}^{(\lambda)} \big(\gamma',\widehat{f}\big)$ by the argument in~(\ref{Hornswaggle}).  For the third term in the second inequality of~(\ref{Hungary}), we have bounded $\overline{U}^{(\lambda)} \big(\gamma',\widehat{f}\big)$ with inequality~(\ref{Nolbach}) and have used that   $  \sup_{\substack{ \lambda<1 \\ \rho\leq \sqrt{2l} }  } \widehat{T}_{\lambda}(\gamma,\gamma' )$ has a Gaussian tail in $\rho'$ for $\gamma'=(\rho',\epsilon')$.


\end{proof}

\section{Linking the original and the Freidlin-Wentzell dynamics}\label{SecTheLink}

Define the linear map $ \widehat{U}^{(\lambda)}:B(\Sigma)\rightarrow B(\Gamma_{V})$ to act as $\widehat{U}^{(\lambda)}f=\widehat{U^{(\lambda)}f}$, where the map $\,\widehat{\text{  }}:B(\Sigma)\rightarrow B(\Gamma_{V})$ is from Part (3) of Definition~\ref{Definitions}.  The lemma below   states that $\widehat{U}^{(\lambda)}f$ satisfies the same integral equation as  $\overline{U}^{(\lambda)}\widehat{f}$ in Part (1) of Lemma~\ref{Freidlin} with an  error term for which Part (3) gives a bound. 

\begin{lemma}\label{ErrorEquation}  Let $f\in B(\Sigma)$ be non-negative.  

\begin{enumerate} 
 \item
  $\widehat{U}^{(\lambda)}f:\Gamma_{V}\rightarrow \R^{+}$ satisfies the  equation
$$
\widehat{f}(\gamma)  = \widehat{h}(\gamma)\widehat{U}^{(\lambda)}\big(\gamma,f\big)+
\int_{\Gamma_{V}}d\gamma'\, \widehat{\mathcal{J}}_{\lambda}(\gamma,\gamma' )\Big(\widehat{U}^{(\lambda)}\big(\gamma,f\big)-\widehat{U}^{(\lambda)}\big(\gamma',f\big)   \Big) +\mathbf{E}_{\lambda}(\gamma),$$
where $\mathbf{E}_{\lambda}(\gamma)$ has the form
\begin{align}\label{Farse}
\mathbf{E}_{\lambda}(\gamma)= \int_{\Gamma_{V}}d\gamma'\, \widehat{\mathcal{J}}_{\lambda}( \gamma,\gamma')\widehat{U}^{(\lambda)}\big(\gamma',f\big)  - \big(  \widehat{ \mathcal{J}_{\lambda}U^{(\lambda)}f } \big) (\gamma)     ,
\end{align}
and the operator $ \mathcal{J}_{\lambda}$ acts on $B(\Sigma)$ with kernel density $\mathcal{J}_{\lambda}(p,p')$. The above is equivalent to the statement
\begin{align}\label{NightengaleII}
\widehat{U}^{(\lambda)}f= \overline{U}^{(\lambda)}\big(\widehat{f}- \mathbf{E}_{\lambda}   \big).
\end{align}

\item  There is a $C>0$ such that for all $\lambda<1$, $f\in B(\Sigma)$, and $H(x,p)>l$
$$ \big| U^{(\lambda)}\big((x,p),\,f) -\widehat{U}^{(\lambda)}\big(\gamma(x,p),f\big)\big|\leq C\,\textup{max}\big(\frac{1}{1+|p|},\lambda\big)\, \widehat{U}^{(\lambda)}\big(\gamma(x,p),f\big) . $$

\item For $\gamma=(\rho,\epsilon)\in \Gamma_{V}$, let $A_{ \gamma }\subset \Gamma_{V}$ be the set 
\begin{align*}
A_{\gamma}=\left\{  \begin{array}{cc} \big\{(\rho',\epsilon') \,\big|\,  \frac{1}{2}\rho \leq \rho'\leq 2\rho \text{ and }\, \epsilon'=\epsilon   \big\} & \rho> \sqrt{2l} \\ \emptyset  & \rho \leq  \sqrt{2l}  ,
    \end{array} \right.  
\end{align*}
and define the function $\mathcal{M}_{\lambda}:\Gamma_{V}^{2}\rightarrow \R^{+}   $
$$\mathcal{M}_{\lambda}(\gamma,\gamma')  =\textup{max}\big(\frac{1}{1+\rho'}, \lambda\big)\Big(\frac{1}{\rho^{2}}\big(1+(\rho'-\rho)^{2}+\lambda^{2}\rho^{2}\big) \chi\big(\gamma'\in  A_{ \gamma}   \big)+\chi\big(\gamma'\notin A_{\gamma}\big)  \Big).  $$
The error $\mathbf{E}_{\lambda}$ from Part (1) is a sum of parts $\mathbf{E}_{\lambda}'$ and $\mathbf{E}_{\lambda}-\mathbf{E}_{\lambda}'$  for which 
there is a $C>0$ such that for all $\lambda<1$, $f\in B(\Sigma)$,  $\gamma=(\rho,\epsilon)\in \Gamma_{V}$,
\begin{align*}
\big|\mathbf{E}_{\lambda}(\gamma)- \mathbf{E}_{\lambda}'(\gamma)\big| & \leq C\,e^{-\rho}\Big(\sup_{H'> \frac{1}{2}\lambda^{-2} } f(s')+\int_{ H'\leq \frac{1}{2}\lambda^{-2} }ds'\,f(s') \Big),\\ \big| \mathbf{E}_{\lambda}'(\gamma)\big| & \leq C\int_{\Gamma_{V}}d\gamma'\, \widehat{\mathcal{J}}_{\lambda}(\gamma,\gamma' )\mathcal{M}_{\lambda}(\gamma,\gamma')\,\widehat{U}^{(\lambda)}\big(\gamma',f\big). \\ 
\end{align*}

\end{enumerate}

\end{lemma}

\begin{proof}\text{  }\\
\noindent Part (1):\hspace{.2cm} Consider the altered escape and jump rates given by the following
\begin{eqnarray*}
 \mathcal{E}'_{\lambda}(x,p) &= &  \max_{x'\in \mathbb{T} }\mathcal{E}_{\lambda}\big(\sqrt{2H(x,p)-2V(x')}\big), \\  \mathcal{J}_{\lambda}^{(x),\prime}(p,p') &= &   \mathcal{J}_{\lambda}(p,p')  +\delta_{0}\big(p'-p\big) \Big(  \mathcal{E}_{\lambda}'(x,p)- \mathcal{E}_{\lambda}\big(\sqrt{2 H(x,p)-2 V(x)}\big)   \Big).
\end{eqnarray*}
With the above constructions $\int_{\R^{+}}dp'  \mathcal{J}_{\lambda}^{(x),\prime}(p,p')=  \mathcal{E}_{\lambda}'(x,p)$ for each $x\in \mathbb{T}$.  Replacing the jump rates  $ \mathcal{J}_{\lambda}  $ by $ \mathcal{J}_{\lambda}^{(x),\prime}$ makes no difference for the underlying process, since it merely adds a spatially-dependent rate of vacuous jumps $p\rightarrow p$  so that the escape rate is invariant of the Hamiltonian evolution. 

Let $\tau$ be a mean-$1$ exponential time and $t_{1}$ be the first ``collision" time according to our new jump rates.  By considering the stopping time $\textup{min}(\tau, t_{1})$,  we are lead to the integral equation
\begin{multline*}
 U^{(\lambda)}\big((x,p),\,f\big) = \frac{f(x,p)  }{1+ \mathcal{E}_{\lambda}'(x,p) }  +\big(1-h(x,p) \big)  \int_{0}^{\infty}dt\,e^{ - t-t \mathcal{E}_{\lambda}'(x,p)      }U^{(\lambda)}\big((\mathbf{x}_{t},\mathbf{p}_{t}),\,f\big)
 \\
\hspace{2cm}+ \int_{0}^{\infty}dt\,e^{ - t-t \mathcal{E}_{\lambda}'(x,p)      }\int_{\R}dp'  \mathcal{J}_{\lambda}^{(\mathbf{x}_{t}),\prime}(\mathbf{p}_{t},p')  U^{(\lambda)}\big((\mathbf{x}_{t},p'),\,f\big) ,
\end{multline*}
where $(\mathbf{x}_{t},\mathbf{p}_{t})$ is the phase space point at time $t$ when evolving according the Hamiltonian $H$ starting from the point $(x,p)$.  The above equation can be reshuffled to give
\begin{multline}\label{Hornets}
f(x,p)    = h(x,p) \big(1+\mathcal{E}_{\lambda}'(x,p)\big)  \int_{0}^{\infty}dt\,e^{ - t-t \mathcal{E}_{\lambda}'(x,p)      }U^{(\lambda)}\big((\mathbf{x}_{t},\mathbf{p}_{t}),\,f\big)  
 \\
+  \Big( U^{(\lambda)}\big((x,p),\,f\big)-\big(1+ \mathcal{E}_{\lambda}'(x,p)\big) \int_{0}^{\infty}dt\,e^{ - t-t\mathcal{E}_{\lambda}'(x,p)      }U^{(\lambda)}\big((\mathbf{x}_{t},\mathbf{p}_{t}),\,f\big) \Big)  \\
+ \big(1+ \mathcal{E}_{\lambda}'(x,p)\big)\int_{0}^{\infty}dt\,e^{ - t-t \mathcal{E}_{\lambda}'(x,p)      }\int_{\R}dp' \mathcal{J}_{\lambda}^{(\mathbf{x}_{t} ),\prime}(\mathbf{p}_{t},p') \Big( \big(U^{(\lambda)}\big((\mathbf{x}_{t},\mathbf{p}_{t}),\,f\big)-\big(U^{(\lambda)}f\big)\big(\mathbf{x}_{t},p'\big) \Big).
\end{multline}
The jump rates $ \mathcal{J}_{\lambda}^{(\mathbf{x}_{t} ),\prime}$ in the last term can be replaced by the original rates $\mathcal{J}_{\lambda}$, since the difference is merely the vacuous jumps.  Moreover, by integrating both sides over $(x,p)\in \Sigma$ against $\delta\big(2^{\frac{1}{2}}H^{\frac{1}{2}}(x,p)-\rho\big)\chi\big(n(x,p)=\epsilon\big)$, we obtain 
\begin{multline*}
\widehat{f}(\gamma)  = \widehat{h}(\gamma)\widehat{U}^{(\lambda)}\big(\gamma,f\big)+
\int_{\Gamma_{V}}d\gamma'\, \widehat{\mathcal{J}}_{\lambda}(\gamma,\gamma' )\Big(\big(\widehat{U}^{(\lambda)}\big(\gamma,f\big)-\big(\widehat{U}^{(\lambda)}\big(\gamma',f\big)    \Big)\\ +\Big(  \int_{\Gamma_{V}}d\gamma'\, \widehat{\mathcal{J}}_{\lambda}(\gamma,\gamma' )\big(\widehat{U}^{(\lambda)}\big(\gamma',f\big)  - \big(  \widehat{ \mathcal{J}_{\lambda}U^{(\lambda)}f } \big) (\gamma)     \Big),
\end{multline*}
for $\gamma=(\rho,\epsilon)$.  We will illustrate the computation for the first term on the right side of~(\ref{Hornets}):
\begin{align*}
\int_{\Sigma} &dxdp\,   \delta\big(2^{\frac{1}{2}}H^{\frac{1}{2}}(x,p)-\rho\big)\chi\big(n(x,p)=\epsilon\big) \big(1+\mathcal{E}_{\lambda}'(x,p)\big)  \int_{0}^{\infty}dt\,e^{ - t-t \mathcal{E}_{\lambda}'(x,p)      }U^{(\lambda)}\big((\mathbf{x}_{t},\mathbf{p}_{t}),\,f\big)\\ & =  \int_{0}^{\infty}dt\, \int_{\Sigma}    dxdp\,\delta\big(2^{\frac{1}{2}}H^{\frac{1}{2}}(x,p)-\rho\big)\chi\big(n(x,p)=\epsilon\big) \big(1+\mathcal{E}_{\lambda}'(x,p)\big) e^{ - t-t \mathcal{E}_{\lambda}'(x,p)      }U^{(\lambda)}\big((x,p),\,f\big)\\  &= \int_{\Sigma}   dxdp\, \delta\big(2^{\frac{1}{2}}H^{\frac{1}{2}}(x,p)-\rho\big)\chi\big(n(x,p)=\epsilon\big) U^{(\lambda)}\big((x,p),\,f\big)=\widehat{U}^{(\lambda)}\big((\rho,\epsilon),\,f\big).
\end{align*}
The first equality uses Fubini's theorem to pull out the integral $\int_{0}^{\infty}dt$, and then employs a change of variables over $\Sigma$ with the dynamical transformation map $(x,p)\rightarrow (\mathbf{x}_{-t},\mathbf{p}_{-t})$ (i.e. backwards time evolution according to the Hamiltonian $H$ for a time  interval of length $t$).   Thus $(\mathbf{x}_{t},\mathbf{p}_{t})$ maps to $(x,p)$, and other expressions do not change since $H^{\frac{1}{2}},\rho, \mathcal{E}_{\lambda}'$ are functions of the energy.  The second equality uses Fubini to commute the integral $\int_{0}^{\infty}dt$ in again, and then computes the integral  $\int_{0}^{\infty}dte^{ - t-t \mathcal{E}_{\lambda}'(x,p)      }=\big(1+\mathcal{E}_{\lambda}'(x,p)\big)^{-1} $.

The equality $\widehat{U}^{(\lambda)}f= \overline{U}^{(\lambda)}\big(\widehat{f}- E_{\lambda}   \big)$ follows  from Part (1) of Lemma~\ref{Freidlin}.  

\vspace{.5cm}

\noindent Part (2):  As in Part (1), let $t_{1}$ be the first ``collision" time with the vacuous jumps included.   If the particle begins at $(x,p)$ with $H(x,p)>l$, then the final time $R$ from Part (1) of Proposition~\ref{LifeOperator} can not occur over the interval $[0,t_{1}]$, since the modulating function $h$  has support on the set $H(x,p)\leq l$.
The value $U^{(\lambda)}\big((x,p),\,f\big)$ can be written as
$$
U^{(\lambda)}\big((x,p),\,f\big)=  \int_{ \Sigma  }\kappa_{(x,p)}^{(\lambda)}(dx'dp')     \int_{\R   }dp'' \frac{\mathcal{J}_{\lambda}^{(x'),\prime}(p',p'')}{\mathcal{E}_{\lambda}'(x',p') }U^{(\lambda)}\big((x',p''),\,f\big) , $$
 where the measure $\kappa_{(x,p)}^{(\lambda)}$ is supported on the set of $(x',p')$ with $\gamma(x',p')=\gamma(x,p)\in \Gamma_{V}$ and is defined by
$$\kappa_{(x,p)}^{(\lambda)}(dx'dp')=\mathbb{E}^{(\lambda)}_{(x,p)}\big[ \delta_{(x',p')}\big(X_{t_{1}},P_{t_{1}^{-}}    \big)    \big].$$
On the other hand,  
\begin{align}\label{Obv}
\widehat{U}^{(\lambda)}\big(\gamma(x,p),f\big)= \int_{ \Sigma  }\kappa_{\gamma(x,p) }(dx'dp')     \int_{\R   }dp'' \frac{\mathcal{J}_{\lambda}^{(x'),\prime}(p',p'')}{\mathcal{E}'_{\lambda}(x',p') } U^{(\lambda)}\big((x',p''),f\big),
\end{align}
where $\kappa_{\gamma }$ is the normalized measure from Definition~\ref{Definitions}.  By the bounds on the escape rates in Part (1) of Proposition~\ref{NewStuff}, the random variable $t_{1}$ is exponential with mean $\geq c\, \textup{min}\big(1, (\lambda  |p|)^{-1}\big)$ for some $c>0$ and all $\lambda<1$ and $p$.  Since the particle is traveling with velocity $p$, it will typically revolve around the level curve on the order of $\textup{min}(|p|, \lambda^{-1})$ times before $t_{1}$ occurs.  The Radon-Nikodym derivative  $\frac{d\kappa_{(x,p)}^{(\lambda)}}{ d\kappa_{\gamma(x,p) }}$ satisfies
\begin{align}\label{Radon}
\sup_{s'\in \Sigma}\Big|\frac{d\kappa_{(x,p)}^{(\lambda)}}{ d\kappa_{\gamma(x,p) } }(s')-1\Big|\leq c'\textup{max}\big(|p|^{-1}, \lambda \big)  
\end{align}
for some $c'$.  Thus,
\begin{align*}
\Big|U^{(\lambda)}\big((x,p)&,\,f\big)-\widehat{U}^{(\lambda)}\big(\gamma(x,p),f\big)\big|\\ &\leq  \int_{ \Sigma  }\kappa_{\gamma(x,p) }(dx'dp')  \Big|\frac{d\kappa_{(x,p)}^{(\lambda)}}{ d\kappa_{\gamma(x,p) }}(x',p') -1\Big|     \int_{\R   }dp'' \frac{\mathcal{J}_{\lambda}^{(x'),\prime}(p',p'')}{\mathcal{E}_{\lambda}'(x',p')} U^{(\lambda)}\big((x',p''),f\big)\\ &\leq c'\,\textup{max}\big(|p|^{-1}, \lambda\big) \,\widehat{U}^{(\lambda)}\big(\gamma(x,p),f\big),
\end{align*}
where we have applied the inequality~(\ref{Radon}) and used the formula~(\ref{Obv}) for $\widehat{U}^{(\lambda)}\big(\gamma(x,p),\,f\big)$.

\vspace{.5cm}

\noindent Part (3):  The expression $\mathbf{E}_{\lambda}(\gamma)$ can be written
\begin{align*}
\mathbf{E}_{\lambda}(\gamma) &=\int_{\Sigma} \kappa_{\gamma}(dxdp) \int_{\R}dp'\, \mathcal{J}_{\lambda}(p,p')\Big(  U^{(\lambda)}\big((x,p') ,\,f\big) - \widehat{U}^{(\lambda)}\big(\gamma(x,p'),\,f  \big)      \Big)\\ & = \int_{\Gamma_{V} }d\gamma' \widehat{\mathcal{J}}_{\lambda}(\gamma,\gamma' ) \int_{\Sigma} \eta_{\gamma'}(dx'dp') \int_{\Sigma} \kappa_{\gamma}(dxdp)\, \delta_{0}(x-x') \\  & \hspace{2cm}  \cdot\frac{  \mathcal{J}_{\lambda}(p,p')      }{ \widehat{\mathcal{J}}_{\lambda}(\gamma,\gamma')   }\Big(  U^{(\lambda)}\big((x,p') ,\,f\big)  -  \widehat{U}^{(\lambda)}\big(\gamma(x,p'),\,f  \big)    \Big).       
\end{align*}
   The second equality is by commuting integrals and using (2) of Remarks~\ref{Remark}:  
$$ \int_{\R}dp'=\int_{\Sigma}ds'\,\delta_{0}(x-x')=  \int_{\Gamma_{V} }d\gamma'  \int_{\Sigma} \eta_{\gamma'}(dx'dp') \delta_{0}(x-x').  $$   
  We define $\mathbf{E}_{\lambda}'(\gamma)$ to be the analogous expression with  the integration $ \int_{\Gamma_{V} }d\gamma'$ replaced by the restricted integration $\int_{\rho'> l }d\gamma'$.  The value $\mathbf{E}_{\lambda}'(\gamma) $ is bounded by 
\begin{multline}\label{Whores}
\big|\mathbf{E}_{\lambda}'(\gamma)\big| \leq  \int_{\rho'>\sqrt{2l} }d\gamma' \widehat{\mathcal{J}}_{\lambda}( \gamma,\gamma')         \sup_{s,s'\in\mathbf{G}_{V}^{-1}(\gamma') }   \Big|  U^{(\lambda)}\big(s,f\big)-U^{(\lambda)}\big(s',f\big) \Big| \\ \int_{\Sigma} \eta_{\gamma'}(dx'dp') \Big| \int_{\Sigma} \kappa_{\gamma}(dxdp)\, \delta_{0}(x-x')  \frac{  \mathcal{J}_{\lambda}(p,p')      }{ \widehat{\mathcal{J}}_{\lambda}( \gamma,\gamma')   }         -1\Big| . 
\end{multline}

For $\mathbf{E}_{\lambda}(\gamma)-\mathbf{E}_{\lambda}'(\gamma)$, notice that
\begin{align}\label{Balk}
\big|\mathbf{E}_{\lambda}(\gamma)-\mathbf{E}_{\lambda}'(\gamma)\big| & \leq \int_{\Sigma} \kappa_{\gamma}(dxdp) \int_{H(x,p')\leq l}dp'\, \mathcal{J}_{\lambda}(p,p') \Big(U^{(\lambda)}\big((x,p') ,\,f\big) + \widehat{U}^{(\lambda)}\big(\gamma(x,p'),\,f  \big)      \Big)\nonumber \\ & \leq l^{\frac{1}{2}}\Big(\sup_{\substack{(x,p)\in\mathbf{G}_{V}^{-1}(\gamma)\\ H(x,p')\leq l } }\mathcal{J}_{\lambda}(p,p')\Big)\int_{\mathbb{T} } dx\int_{H(x,p')\leq l}dp'\nonumber \\ & \hspace{4cm}  \cdot \Big(U^{(\lambda)}\big((x,p') ,\,f\big) + \widehat{U}^{(\lambda)}\big(\gamma(x,p'),\,f  \big)      \Big)\nonumber \\  & \leq Ce^{-\rho}\int_{H'\leq l}ds'\,U^{(\lambda)}\big(s' ,f\big) \nonumber \\ &  \leq  C e^{-\rho}\Big( \sup_{H'> \frac{1}{2}\lambda^{-2} } f(s')+\int_{ H'\leq \frac{1}{2}\lambda^{-2} }ds'\,f(s')     \Big),
\end{align}
where $H=H(s)$ and $H'=H(s')$.  The second inequality uses (3) of Remarks~\ref{Remark} and that for $\rho^{2}>l=1+2\sup_{x}V(x)$
$$\frac{ \big(\rho^{2}-2V(x')\big)^{-\frac{1}{2} }   }{ \int_{\mathbb{T} }dx\, \big(\rho^{2}-2V(x)\big)^{-\frac{1}{2} }   } \leq \frac{ \sup_{x}\big(\rho^{2}-2V(x)\big)^{-\frac{1}{2} }   }{ \inf_{x}\, \big(\rho^{2}-2V(x)\big)^{-\frac{1}{2} }   } \leq l^{\frac{1}{2}}.  $$
 The supremum over the values of $\mathcal{J}_{\lambda}(p,p')$ in the second line of~(\ref{Balk}) decays super-exponentially for large $\rho$, and we bounded it by a multiple of $e^{-\rho}$.   For the third inequality in~(\ref{Balk}), we have also used that  $\int_{\mathbb{T} } dx\int_{H(x,p')\leq l}dp'=  \int_{H\leq l}ds$, and  the fourth inequality is by Part (2) of Lemma~\ref{LemPick}.

The following two statements hold, where (I) is by Part (2), and  we prove (II)  below.  

\begin{enumerate}[(I).]

\item  There is a $C>0$ such that for all $\lambda<1$, $\gamma'=(\rho',\epsilon')\in\Gamma_{V}$ with $\rho'>\sqrt{2l}$, and non-negative $f\in B(\Gamma_{V})$,
 $$\sup_{s,s'\in\mathbf{G}_{V}^{-1}(\gamma') }   \Big|  U^{(\lambda)}\big(s,f\big)-U^{(\lambda)}\big(s',f\big) \Big| \leq C\,\textup{max}\big(\frac{1}{1+\rho'}, \lambda\big)\, \widehat{U}^{(\lambda)}\big(\gamma',f  \big).
$$

\item There is a $C>0$ such that for all $\lambda\leq 1$ and $\gamma, \gamma'\in \Gamma_{V}$,
\begin{multline*}
 \int_{\Sigma} \eta_{\gamma'}(dx'dp') \Big| \int_{\Sigma} \kappa_{\gamma}(dxdp)\, \delta_{0}(x-x')  \frac{  \mathcal{J}_{\lambda}(p,p')      }{ \widehat{\mathcal{J}}_{\lambda}(\gamma,\gamma' )   }         -1\Big|   \\ \leq \chi\big( \gamma'\notin A_{\gamma}  \big) +\frac{C}{\rho^{2} }\big(1+|\rho'-\rho|^{2}+\lambda^{2}\rho^{2}\big)\chi\big( \gamma'\in A_{\gamma}  \big) .
\end{multline*}

\end{enumerate}
For the right side of (I), we have used that $\rho'(x,p)=(2H(x,p))^{\frac{1}{2}}$ is close to $|p|$ for $|p|\gg 1$.   

Given (I) and (II) then  by~(\ref{Whores}) 
\begin{multline*}
\big|\mathbf{E}_{\lambda}'(\gamma)\big| \leq 2C \int_{\Gamma_{V} }d\gamma' \,\widehat{\mathcal{J}}_{\lambda}(\gamma,\gamma' ) \widehat{U}^{(\lambda)}\big(\gamma',f  \big)\\ \textup{max}\big(\frac{1}{1+\rho'}, \lambda\big)\, \Big(\chi\big( \gamma'\notin A_{\gamma}  \big) +\frac{C}{\rho^{2} }\big(1+|\rho'-\rho|^{2}+\lambda \rho^{2}\big)\chi\big( \gamma'\in A_{\gamma}  \big)   \Big).
\end{multline*}
   This would complete the proof.

Now we will prove (II).  Notice that the expression in (II) is $\leq 2$ for all $\gamma, \gamma'$, since  by  definition of $\widehat{\mathcal{J}}_{\lambda}$ 
$$ \int_{\Sigma} \eta_{\gamma'}(dx'dp')  \int_{\Sigma} \kappa_{\gamma}(dxdp)\, \delta_{0}(x-x')  \frac{  \mathcal{J}_{\lambda}(p,p')      }{ \widehat{\mathcal{J}}_{\lambda}(\gamma,\gamma' )   }    =1. $$       For $\gamma=(\rho,\epsilon)$ with $\rho\gg 1$,  the following bounds hold.

\begin{enumerate}[(i).]
\item There is a $c>0$ such that for all $\gamma$ with $\rho> \sqrt{2l}$,
$$\sup_{x',x''\in \mathbb{T} }\Big|\int_{\Sigma} \kappa_{\gamma}(dxdp)\, \Big(\delta_{0}(x-x')- \delta_{0}(x-x'')\Big)\Big|\leq  \frac{c}{(1+\rho)^{2}} . $$

\item There is a $c>0$ such that for all $\gamma_{1}\in \Gamma_{V}$ and $\gamma_{2}\in A_{\gamma_{1}}$,
$$ \sup_{ \substack{  \gamma_{1}=\gamma(x,p_{1})=\gamma(x',p_{1}') \\ \gamma_{2}=\gamma(x,p_{2})=\gamma(x',p_{2}')   }    } \Big|    \mathcal{J}_{\lambda}(p_{1},p_{2}) - \mathcal{J}_{\lambda}(p_{1}',p_{2}')\Big|\leq   c\frac{1}{\rho^{2}_{1} }\big(1+\big| \rho_{2}-\rho_{1}    \big|^{2}+\lambda^{2}\rho_{1}^{2}\big)\,\inf_{ \substack{  \gamma_{1}=\gamma(x,p_{1}) \\ \gamma_{2}=\gamma(x,p_{2})  } } \mathcal{J}_{\lambda}(p_{1},p_{2}).   $$
\end{enumerate}
Statement (i) concerns only the level curves of the Hamiltonian.  The value  $\int_{\Sigma} \eta_{\gamma}(dxdp)\, \delta_{0}(x-x')$ is the density (normalized to one) for the amount of time that the particle will spend at the point $x'$ when revolving once around the level curve $\gamma$.  By (3) of Remarks~\ref{Remark}, we have the following closed formula for $\rho>\sqrt{2l}$
$$\int_{\Sigma} \kappa_{\gamma}(dxdp)\, \delta_{0}(x-x')=\frac{ \big(\rho^{2}-2V(x')\big)^{-\frac{1}{2} }   }{ \int_{\mathbb{T}}dx\, \big(\rho^{2}-2V(x)\big)^{-\frac{1}{2} }   }.  $$
  Thus (i) is smaller than 
\begin{multline*}
\sup_{x',x''\in \mathbb{T} }\Big|\int_{\Sigma} \kappa_{\gamma}(dxdp)\, \Big(\delta_{0}(x-x')- \delta_{0}(x-x'')\Big)\Big|\leq \frac{2 \sup_{x}V(x)    }{ \rho^{2}-2\sup_{x}V(x)   }\leq \frac{4\sup_{x}V(x)}{\rho^{2}},
\end{multline*}
where the inequalities have used the restriction $\rho^{2}> 2l > 4\sup_{x}V(x)$.

For statement (ii), let $\gamma_{1}\in \Gamma_{V}$ and $\gamma_{2}\in A_{\gamma_{1}}$. For $0\leq V \leq 2\sup_{x}V(x) $,   
\begin{multline}\label{DerBnd}
\Big| \frac{d}{dV}\mathcal{J}_{\lambda}\big(  \sqrt{\rho_{1}^{2}-V} ,\sqrt{\rho_{2}^{2}-V}  \big)\Big|\leq \Big|\frac{1}{\sqrt{\rho_{2}^{2}-V}}-\frac{1}{\sqrt{\rho_{1}^{2}-V}}  \Big| e^{-\frac{1}{2}\big(\frac{1+\lambda}{2} \sqrt{\rho_{2}^{2}-V}-\frac{1-\lambda}{2}\sqrt{\rho_{1}^{2}-V}    \big)^{2}        }   \\ +\big|\rho_{2}-\rho_{1}   \big|\,\Big|\frac{1}{\sqrt{\rho_{2}^{2}-V}}-\frac{1-\lambda}{1+\lambda}\frac{1}{\sqrt{\rho_{1}^{2}-V}}  \Big|\,\Big| \sqrt{\rho_{2}^{2}-V}-\frac{1-\lambda}{1+\lambda}\sqrt{\rho_{1}^{2}-V}    \Big| e^{-\frac{1}{2}\big(\frac{1+\lambda}{2} \sqrt{\rho_{2}^{2}-V}-\frac{1-\lambda}{2}\sqrt{\rho_{1}^{2}-V}    \big)^{2}        }  \\ \leq c'  \frac{1}{\rho_{1}^{2}  }\big(1+|\rho_{2}-\rho_{1}|^{2}+\lambda^{2}\rho_{1}^{2}\big)      \mathcal{J}_{\lambda}\big(  \sqrt{\rho_{1}^{2}-V} ,\sqrt{\rho_{2}^{2}-V} \big).
\end{multline}

With the bound of the derivative relation~(\ref{DerBnd}), it follows that for any $0\leq V,V'\leq 2\sup_{x}V(x)$
\begin{multline*}
\Big|\mathcal{J}_{\lambda}\big( \sqrt{\rho_{1}^{2}-V'},\sqrt{\rho_{2}^{2}-V'}  \big)- \mathcal{J}_{\lambda}\big(  \sqrt{\rho_{1}^{2}-V},\sqrt{\rho_{2}^{2}-V} \big) \Big| \\ \leq e^{  2c'   }\frac{1}{\rho_{1}^{2}   }\big(1+|\rho_{2}-\rho_{1}|^{2}+\lambda^{2}\rho_{1}^{2}\big)  \mathcal{J}_{\lambda}\big(  \sqrt{\rho_{1}^{2}-V} ,\sqrt{\rho_{2}^{2}-V}  \big)   ,
\end{multline*}
where we have used that $\frac{1}{\rho_{1}^{2}   }\big(1+|\rho_{2}-\rho_{1}|^{2}+\lambda^{2}\rho_{1}^{2}\big)     \leq 2 $ by our constraints on $\rho_{1},\rho_{2}$.  This proves (ii) with $c=e^{c'}$.

By the triangle inequality and supremizing over everything
\begin{align*}
\sup_{\gamma'=\gamma(x',p')    }\Big| &\int_{\Sigma} \kappa_{\gamma}(dxdp)\, \delta_{0}(x-x')  \frac{  \mathcal{J}_{\lambda}(p,p')      }{ \widehat{\mathcal{J}}_{\lambda}( \gamma,\gamma')   }         -1\Big|\\ & \leq  \sup_{x',x''\in \mathbb{T} }\Big|\int_{\Sigma} \kappa_{\gamma}(dxdp)\, \Big(\delta_{0}(x-x')- \delta_{0}(x-x'')\Big)\Big|\,\Big(\sup_{ \substack{  \gamma=\gamma(x,p_{1}) \\ \gamma'=\gamma(x,p_{2})  } } \frac{\mathcal{J}_{\lambda}(p_{1},p_{2}) }{  \widehat{\mathcal{J}}_{\lambda}(\gamma,\gamma' )          } \Big)\\ &+\Big( \sup_{ \substack{  \gamma=\gamma(x,p_{1})=\gamma(x',p_{1}') \\ \gamma'=\gamma(x,p_{2})=\gamma(x',p_{2}')   }    } \Big|    \mathcal{J}_{\lambda}(p_{1},p_{2}) - \mathcal{J}_{\lambda}(p_{1}',p_{2}')\Big| \Big)\Big(\sup_{x'\in \mathbb{T} }\int_{\Sigma} \kappa_{\gamma}(dxdp)\, \delta_{0}(x-x') \Big)\\ &\leq \frac{2c}{\rho^{2}   }+\frac{2c}{\rho^{2} }\big(1+|\rho'-\rho|^{2}+\lambda^{2}\rho^{2}\big)      \leq \frac{4c}{ \rho^{2} }\big(1+|\rho'-\rho|^{2}+\lambda^{2}\rho^{2}\big)       .
\end{align*}

\end{proof}

\begin{lemma}\label{LemHorseShoe}
There is a $C>0$ such that for all non-negative $f\in B(\Sigma)$  and $\lambda<1$,
\begin{eqnarray*}
\sup_{H\geq \frac{1}{2}\lambda^{-2} } U^{(\lambda)} \big(s,f\big)& \leq & C\lambda^{-1}\sup_{H> \frac{1}{2}\lambda^{-2} } f(s) +\sup_{H\leq \frac{1}{2}\lambda^{-2}   } U^{(\lambda)} \big(s,f\big)\text{, or}\\
\sup_{H> \frac{1}{2}\lambda^{-2} } U^{(\lambda)} \big(s,f\big)& \leq & C\lambda^{-1}\sup_{H\geq \frac{1}{2}\lambda^{-2} } f(s) +\sup_{l<H\leq \frac{1}{2}\lambda^{-2}   } U^{(\lambda)} \big(s,f\big)+Ce^{-\lambda^{-1}}\|f\|_{\infty},     
\end{eqnarray*}
where $H=H(s)$.

\end{lemma}

\begin{proof}
For $s\in \Sigma$ with $H(s)\geq \frac{1}{2}\lambda^{-2}$, we can write 
 $U^{(\lambda)}\big(s,f\big)$ as
\begin{align}
U^{(\lambda)}\big(s,f\big)&=\mathbb{E}_{s}^{(\lambda)}\Big[\int_{0}^{\infty}dr e^{-\int_{0}^{r}dv\,h(S_{v})  } \,f(S_{r})  \Big]\nonumber  \\ \label{Jazz}  & 
\leq \|f\|_{\infty}\mathbb{E}^{(\lambda)}_{s}[ \omega   ] + \mathbb{E}_{s}^{(\lambda)}\Big[U^{(\lambda)}\big(S_{\omega},f\big)\Big]  \\ \label{Jizz}
& \leq C\lambda^{-1}\|f\|_{\infty}+ \sup_{l<H'\leq \frac{1}{2}\lambda^{-2}} U^{(\lambda)}\big(s',f\big)+\mathbb{E}_{s}^{(\lambda)}\Big[U^{(\lambda)}\big(S_{\omega},f\big)\chi\big(H(S_{\omega})\leq l  \big)\Big],
\end{align} 
 where $\omega$ is the hitting time that $H(S_{r})$ jumps below $\frac{1}{2}\lambda^{-2}$.  For the second inequality, we have used that $\mathbb{E}^{(\lambda)}_{s}[ \omega   ]$ is bounded by a constant multiple of $\lambda^{-1}$ for all $\lambda<1$ and $s$ with $H(s)\geq \frac{1}{2}\lambda^{-2}$, which we show below.  

 Let the function $W:\Sigma\rightarrow [0,1]$ be defined as in Part (3) of Proposition~\ref{NewStuff}.  By Part (3) of Proposition~\ref{NewStuff}, 
there is a $c>0$ such that the process $W(S_{t})+c\lambda t$ is a supermartingale over the time interval $[0,\omega ]  $.   Thus, the optional stopping theorem gives the first inequality below.
\begin{align*}
 \mathbb{E}_{s}^{(\lambda)}[\omega] \leq \frac{1}{c\lambda} \mathbb{E}_{s}^{(\lambda)}[W(s)-W(S_{\omega})] \leq  \frac{1}{c\lambda}.   
\end{align*}
  For the second inequality, we have  used that $0\leq W\leq 1$.  Thus, $ \mathbb{E}_{s}^{(\lambda)}[\omega]$ is bounded by $C\frac{1}{\lambda}$ for $C=\frac{1}{c}$, and plugging this into~(\ref{Jazz}) gives the first inequality in the statement of the lemma.

Now we will show that the last term on the right side of~(\ref{Jizz}) is $\|f\|_{\infty}\mathit{O}(e^{-\lambda^{-1}})$.   Let $\omega'$ be the collision time which precedes $\omega$, and $\beta_{s}^{(\lambda)}(s')$ be the conditional probability density for $s'=S_{\omega}$ when given the value $s=S_{\omega'}$.  By the strong Markov property, 
\begin{align}\label{FratParty}
\mathbb{E}_{s}^{(\lambda)}\Big[U^{(\lambda)}\big(S_{\omega},f\big)\chi\big(H(S_{\omega})\leq l  \big)\Big]  & =  \mathbb{E}_{s}^{(\lambda)}\Big[ \int_{H'\leq l}ds'\,  \beta_{S_{\omega'}}^{(\lambda)}(s') U^{(\lambda)}\big(s',f\big)   \Big] \nonumber  \\ & \leq  \Big(\sup_{ \substack{ H>\frac{1}{2}\lambda^{-2}   \\ H'\leq l }  }  \beta_{s}^{(\lambda)}(s')\Big) \int_{H'\leq l}ds' \,U^{(\lambda)}\big(s',f\big) \nonumber \\ &  \leq c\lambda^{-\frac{1}{2}}\|f\|_{\infty}\sup_{\substack{H>\frac{1}{2}\lambda^{-2}  \\ H'\leq l} }  \beta_{s}^{(\lambda)}(s') ,
\end{align}
where $H=H(s)$, $H'=H(s')$. The second inequality is by Part (2) of Lemma~\ref{LemPick}
and that $\int_{\Sigma}ds\,e^{-\lambda\, H(s) }\propto \lambda^{-\frac{1}{2}}$ for $\lambda\ll 1$.  To make use of~(\ref{FratParty}), we must bound the values of $\beta_{s}^{(\lambda)}(s')$.    Let $\big(\mathbf{x}_{t}(s), \mathbf{p}_{t}(s)\big)$ be the phase space point at time $t>0$ when evolving according to the Hamiltonian evolution from the point    $s\in \Sigma $.    The density $\beta_{s}^{(\lambda)}(s')  $ can be written
\begin{align}\label{Jot}
\beta_{s}^{(\lambda)}(s') = \kappa_{s}^{(\lambda)}(x') \, \frac{\mathcal{J}_{\lambda}(p   , p' )\, \chi\big( H(s')\leq \frac{1}{2}\lambda^{-2} \big)  }{\int_{H(x',p'')\leq  \frac{1}{2}\lambda^{-2}}dp''\,   \mathcal{J}_{\lambda}(p   , p'' ) }
 , 
\end{align}
where $\kappa_{s}^{(\lambda)}$ is the probability measure on $\mathbb{T}$ given by
\begin{align}\label{Tuhut}
\kappa_{s}^{(\lambda)}(x')=  \frac{ \int_{0}^{\infty}dt\,e^{-\int_{0}^{t}dr\,\mathcal{E}_{\lambda}(\mathbf{p}_{r}(s)  )    } \int_{H(\mathbf{x}_{t}(s),p'')\leq \frac{1}{2}\lambda^{-2} }dp'' \mathcal{J}_{\lambda}\big( \mathbf{p}_{t}(s), p'' \big)\, \delta(\mathbf{x}_{t}(s)-x')  }{     \int_{0}^{\infty}dt\,e^{-\int_{0}^{t}dr\,\mathcal{E}_{\lambda}(\mathbf{p}_{r}(s))}  \int_{H(\mathbf{x}_{t}(s),p'')\leq \frac{1}{2}\lambda^{-2} }dp'' \mathcal{J}_{\lambda}\big( \mathbf{p}_{t}(s), p'' \big)    }.  
\end{align}
However, there is a $c>0$ such that for all $\lambda<1$,  
$$ (\text{i}).\hspace{.2cm} \sup_{\substack{ H>\frac{1}{2}\lambda^{-2}\\x'\in \mathbb{T}    }   }   \kappa_{s}^{(\lambda)}(x')< c, \quad \quad  (\text{ii}).\hspace{.2cm}  \sup_{\substack{ H>\frac{1}{2}\lambda^{-2},   \\ H'\leq l  }}\frac{\mathcal{J}_{\lambda}(p   , p' )  }{\int_{H(x',p'')\leq  \frac{1}{2}\lambda^{-2}}dp''\,   \mathcal{J}_{\lambda}(p   , p'' ) }\leq  c e^{-\frac{1}{16}\lambda^{-2}}.         $$
Plugging~(\ref{Jot}) into~(\ref{FratParty}) and using  (i) and (ii) gives the result.  Statement (ii) follows from the Gaussian decay in the rates $\mathcal{J}_{\lambda}(p   , p' )$ and that $V(x)$ is bounded.   For statement (i), we actually have that $\sup_{x'}\big|\kappa_{s}^{(\lambda)}(x')-1|=\mathit{O}(\lambda)$ for $H(s)>\lambda^{-1}$ and  $\lambda\ll 1$.  We will sketch why.  Notice that by the conservation of energy $H(s)=H\big(\mathbf{x}_{t}(s), \mathbf{p}_{t}(s)\big)$.  For $|p|\geq \lambda^{-1}>\sqrt{2\sup_{x}V(x)}$, then 
$\mathbf{p}_{t}(s)$ and $p$ have the same sign and  
$$\big|  \mathbf{p}_{t}(s)-p  \big|= \big| \sqrt{ p^{2}+2V(x)-2V(\mathbf{x}_{t}(s) ) }  -|p|  \big| \leq   \frac{2\sup_{x}V(x)}{ |p|}.  $$
This implies that there is very little deviation of the values  $  \mathbf{p}_{t}(s)$ from $p$ when $|p|>\lambda^{-1}$.  Consequently, $\mathbf{x}_{t}(s)$ revolves around the torus with nearly uniform speed $|p|$, and the terms $\mathcal{E}_{\lambda}(\mathbf{p}_{r}(s)  )$  and $\mathcal{J}_{\lambda}\big( \mathbf{p}_{t}(s), p'' \big)$ in the expression~(\ref{Tuhut}) are approximately equal to $\mathcal{E}_{\lambda}(p )$  and $\mathcal{J}_{\lambda}\big( p, p'' \big)$, respectively.  Moreover, the starting location $x\in \mathbb{T}$ on the torus only makes a difference in $\kappa_{s}^{(\lambda)}(x')$ on the order of $\mathit{O}(\lambda)$, since the decay factor in~(\ref{Tuhut})  satisfies  $e^{-\int_{0}^{t}dr\,\mathcal{E}_{\lambda}(\mathbf{p}_{r}(s)  )    }=1+\mathit{O}(\lambda)$ for times $t$ up to the first revolution around the torus.  This follows because the speed is $\approx |p|$ and     
$ \frac{\mathcal{E}_{\lambda}(p) }{|p|} =\mathit{O}(\lambda  ) $  for $|p|>\lambda^{-1}$ by  Part (1) of Proposition~\ref{NewStuff}.


\end{proof}

Putting Parts (1) and (3) of Lemma~\ref{ErrorEquation} together gives an inequality including the values of the function  $\widehat{U}^{(\lambda)}f$ and  weighted integrals of those values.  This suggests using a Gronwall-type recursive scheme to obtain bounds for  $\widehat{U}^{(\lambda)}f$.  However, it is useful to bring the results of Lemma~\ref{ErrorEquation} into a more tailored form that is amenable to recursion, and this is the purpose of the following lemma.

\begin{lemma}\label{CryBabyZero}
There exist $C, C',L>0$ such that for  all $\lambda<1$, $\gamma=(\rho,\epsilon)$ with $\rho\leq \lambda^{-1}$, and non-negative $f\in B(\Sigma)$ 
\begin{multline}\label{Juicer}
\widehat{U}^{(\lambda)}\big(\gamma,f\big)  \leq   C\|f\|_{\infty}+C\rho\sup_{\rho'>\lambda^{-1} }f(\gamma')+C\int_{\rho'\leq \lambda^{-1}    }d\gamma'\,\big(1+ \textup{min}(\rho, \rho')  \big)\,\widehat{f}(\gamma') \\+ C'\int_{\sqrt{ 2L}\leq \rho'\leq \lambda^{-1}    }d\gamma'\,\frac{1+ \textup{min}(\rho, \rho')  }{ (1+ \rho')^{3}   }\,\widehat{U}^{(\lambda)}\big(\gamma',f\big) ,
\end{multline}
 and $C'\int_{ \sqrt{2L}\leq \rho'\leq \lambda^{-1}    }d\gamma'\frac{1}{(1+\rho')^{2}}  \leq \frac{1}{2}$.

\end{lemma}

\begin{proof}
  By Part (1) of Lemma~\ref{ErrorEquation}, we have the equality $\widehat{U}^{(\lambda)}f= \overline{U}^{(\lambda)}\big(\widehat{f}- \mathbf{E}_{\lambda}   \big)$.  Applying Part (2) of Lemma~\ref{Freidlin} gives the inequality
\begin{align}\label{CryBaby}
&\widehat{U}^{(\lambda)}\big(\gamma,f\big) =\overline{U}^{(\lambda)}\big(\gamma,\widehat{f}- \mathbf{E}_{\lambda}   \big)\leq \overline{U}^{(\lambda)}\big(\gamma,\widehat{f}\big)+ \overline{U}^{(\lambda)}\big(\gamma,|\mathbf{E}_{\lambda}'|\big)+\overline{U}^{(\lambda)}\big(\gamma,|\mathbf{E}_{\lambda}-\mathbf{E}_{\lambda}'|\big)\nonumber  \\  &  \leq c\Big(  \big\|\frac{\widehat{f}}{\widehat{\mathcal{E}}_{\lambda}}\big\|_{\infty}+\big\|\frac{\mathbf{E}_{\lambda}'}{\widehat{\mathcal{E}}_{\lambda}}  \big\|_{\infty}+\big\|\frac{\mathbf{E}_{\lambda}-\mathbf{E}_{\lambda}'}{\widehat{\mathcal{E}}_{\lambda}}  \big\|_{\infty} \Big) +c\rho\sup_{\rho'> \lambda^{-1}}\Big(\frac{\widehat{f}(\gamma')}{\widehat{\mathcal{E}}_{\lambda}(\gamma')   }+\frac{|\mathbf{E}_{\lambda}'(\gamma')|}{\widehat{\mathcal{E}}_{\lambda}(\gamma')   }+\frac{|\mathbf{E}_{\lambda}(\gamma')-\mathbf{E}_{\lambda}'(\gamma')|}{\widehat{\mathcal{E}}_{\lambda}(\gamma')   } \Big)\nonumber \\ & \hspace{1cm}+ c\int_{\rho'\leq \lambda^{-1}   }d\gamma'\,\big(1+ \textup{min}(\rho, \rho')\big)\Big( \frac{\widehat{f}(\gamma')}{\widehat{\mathcal{E}}_{\lambda}(\gamma') }+\frac{|\mathbf{E}_{\lambda}'(\gamma')|}{\widehat{\mathcal{E}}_{\lambda}(\gamma') }+\frac{|\mathbf{E}_{\lambda}(\gamma')-\mathbf{E}_{\lambda}'(\gamma')|}{\widehat{\mathcal{E}}_{\lambda}(\gamma') }\Big) 
\end{align}
for some $c>0$, where  $\mathbf{E}_{\lambda}'$ is defined as in the proof of Part (3) of Proposition~\ref{ErrorEquation}.  We will show that there is $C'>0$  such that
\begin{multline}\label{Protestant}
\widehat{U}^{(\lambda)}\big(\gamma,f\big) \leq  C'\|f\|_{\infty}+    C'\rho\sup_{H'> \frac{1}{2}\lambda^{-1}}f(s')+  C'\lambda \sup_{\rho'\leq \lambda^{-1}}\widehat{U}^{(\lambda)}\big(\gamma',f\big) \\ + C'\int_{\rho'\leq \lambda^{-1}   }d\gamma'\,\big(1+ \textup{min}(\rho, \rho')\big)\widehat{f}(\gamma') +C'\int_{ \rho'\leq \lambda^{-1}   }d\gamma'\,\frac{\big(1+ \textup{min}(\rho, \rho')\big)}{(1+\rho')^{3}}\widehat{U}^{(\lambda)}\big(\gamma',f\big) .  
\end{multline}
Given~(\ref{Protestant}), we can split the integral $\int_{ \rho'\leq \lambda^{-1}    }$ of the last term  into two parts $\int_{ \sqrt{2L}\leq \rho'\leq \lambda^{-1}    }$ and $\int_{\rho''\leq \sqrt{2L} }$  with  $L>0$ large enough so that 
 $$C'\int_{ \sqrt{2L}\leq \rho'\leq \lambda^{-1}    }d\gamma'\frac{1}{(1+\rho')^{2}}  \leq \frac{C'}{\sqrt{2L} }\leq \frac{1}{2}.$$
We can bound the remainder $\int_{\rho''\leq \sqrt{2L} }$ through the inequalities
\begin{align}\label{Hash}
\int_{\rho''\leq \sqrt{2L} }d\gamma''\,\frac{\big(1+ \textup{min}(\rho'', \rho)\big)}{(1+\rho'')^{3}}\,\widehat{U}^{(\lambda)}\big(\gamma'',f\big)  & \leq \int_{\rho''\leq \sqrt{2L} }d\gamma''\,\widehat{U}^{(\lambda)}\big(\gamma'',f\big)=\int_{H\leq L }ds\,U^{(\lambda)}\big(s,f\big)\nonumber
  \\ & \leq  a\| f\|_{\infty}+    a\int_{H\leq \frac{1}{2}\lambda^{-2} }ds\,f(s) \nonumber \\ & =  a\| f\|_{\infty}+    a\int_{\rho\leq \lambda^{-1} }d\gamma\,\widehat{f}(\gamma),
\end{align}
where $H=H(s)$.  The second inequality above holds for some $a>0$ by Part (2) of Lemma~\ref{LemPick}.  Thus with~(\ref{Protestant}) and~(\ref{Hash}), 
 there are $C',C'',L>0$  with $\frac{C'}{L}\leq \frac{1}{2}$ such that
\begin{multline}\label{PostProtestant}
\widehat{U}^{(\lambda)}\big(\gamma,f\big) \leq  C''\|f\|_{\infty}+    C''\rho\sup_{H'> \frac{1}{2}\lambda^{-1}}f(s')+  C''\lambda \sup_{\rho'\leq \lambda^{-1}}\widehat{U}^{(\lambda)}\big(\gamma',f\big) \\ + C''\int_{\rho'\leq \lambda^{-1}   }d\gamma'\,\big(1+ \textup{min}(\rho, \rho')\big)\widehat{f}(\gamma') +C'\int_{L\leq \rho'\leq \lambda^{-1}   }d\gamma'\,\frac{\big(1+ \textup{min}(\rho, \rho')\big)}{(1+\rho')^{3}}\widehat{U}^{(\lambda)}\big(\gamma',f\big) .  
\end{multline}
By supremizing both sides of~(\ref{PostProtestant}) over $\gamma=(\rho,\epsilon)$ with $\rho\leq \lambda^{-1}$, we obtain
$$  \sup_{\rho'\leq \lambda^{-1}}\widehat{U}^{(\lambda)}\big(\gamma',f\big) \leq \frac{C''}{\frac{1}{2}-C''\lambda}\Big(   \|\widetilde{f}\|_{\infty}+    \lambda^{-1}      \sup_{\rho'> \lambda^{-1}}f(\gamma') +  \int_{\rho'\leq \lambda^{-1}   }d\gamma'\,\big(1+ \textup{min}(\rho', \lambda^{-1}) \big)\widehat{f}(\gamma') \Big).  $$
Plugging this bound back into~(\ref{Protestant}) gives the inequality~(\ref{Juicer}) for small enough $\lambda$ ($\lambda$ bounded away from zero does not pose a problem).  

Now, we work to prove~(\ref{Protestant}) starting from~(\ref{CryBaby}).  Since $\widehat{\mathcal{E}}_{\lambda}$ is bounded away from zero by Part (2) of Proposition~\ref{NewStuff}, the expressions on the right side of~(\ref{CryBaby}) with $\widehat{f}$ and $\mathbf{E}_{\lambda}-\mathbf{E}_{\lambda}'$ do not pose a problem, since, in particular, we can bound  $\mathbf{E}_{\lambda}-\mathbf{E}_{\lambda}'$ with Part (3) of Lemma~\ref{ErrorEquation}. For the term $\mathbf{E}_{\lambda}'$, we have the following expressions to bound:
$$ (\text{i}). \hspace{.2cm} \big\|\frac{\mathbf{E}_{\lambda}'}{\widehat{\mathcal{E}}_{\lambda}} \big\|_{\infty}, \quad \quad (\text{ii}).  \hspace{.2cm}  \rho\sup_{\rho'> \lambda^{-1}}\frac{|\mathbf{E}_{\lambda}'(\gamma')|}{\widehat{\mathcal{E}}_{\lambda}(\gamma')   } , \quad \quad   (\text{iii}). \hspace{.2cm}  \int_{\rho'\leq \lambda^{-1}   }d\gamma'\,\big(1+ \textup{min}(\rho, \rho')\big)\frac{|\mathbf{E}_{\lambda}'(\gamma')|}{\widehat{\mathcal{E}}_{\lambda}(\gamma') }.       $$
We will discuss (ii) and (iii), since (i) is handled similarly.  In each case, we seek a bound using a linear combination of the terms on the right side of~(\ref{Protestant}).  
 
 By Part (3) of Lemma~\ref{ErrorEquation} and $\rho\leq \lambda^{-1}$,
\begin{align}\label{Hereafter}
\rho\sup_{\rho> \lambda^{-1}}\frac{|\mathbf{E}_{\lambda}'(\gamma')|}{\widehat{\mathcal{E}}_{\lambda}(\gamma')   } & \leq  C\lambda^{-1}\sup_{\rho\geq \lambda^{-1}}\int_{\Gamma_{V}}d\gamma'\, \widehat{T}_{\lambda}(\gamma,\gamma')\,\mathcal{M}_{\lambda}(\gamma,\gamma')\,\widehat{U}^{(\lambda)}\big(\gamma',f\big)\nonumber \\ & \leq C\lambda^{-1} \Big(\sup_{\rho\geq \lambda^{-1}, \gamma'\in \Gamma_{V}} \widehat{T}_{\lambda}(\gamma,\gamma')\,\mathcal{M}_{\lambda}(\gamma,\gamma')      \Big) \Big(   \sup_{\gamma'}\widehat{U}^{(\lambda)}\big(\gamma',f\big) \Big)\nonumber \\ &\leq 
C' \lambda^{2}\, \sup_{\gamma'}\widehat{U}^{(\lambda)}\big(\gamma',f\big) \nonumber \\
&\leq C''\lambda^{2}\|f\|_{\infty}+ C''\lambda\sup_{H'> \frac{1}{2}\lambda^{-2} } f(s') +C''\lambda^{2}\sup_{\rho'\leq \lambda^{-1}   }  \widehat{U}^{(\lambda)}\big(\gamma',f\big)
. 
\end{align}

The fourth inequality in~(\ref{Hereafter}) follows since
\begin{align}\label{Rabbel}
\sup_{\rho'\geq \lambda^{-1} }    \widehat{U}^{(\lambda)}\big(\gamma',f\big) &\leq  \sup_{H'> \frac{1}{2}\lambda^{-2} }    U^{(\lambda)}\big(s',f\big) \nonumber \\ &\leq ce^{-\lambda^{-1}}\|f\|_{\infty}+c\lambda^{-1}\sup_{H'> \frac{1}{2}\lambda^{-2} } f(s') +\sup_{l<H'\leq \frac{1}{2}\lambda^{-2}   }  U^{(\lambda)}\big(s',f\big)\nonumber  \\ &\leq c\|f\|_{\infty}+ c\lambda^{-1}\sup_{H'> \frac{1}{2}\lambda^{-2} } f(s') +c'\sup_{\rho'\leq \lambda^{-1}   }  \widehat{U}^{(\lambda)}\big(\gamma',f\big),
\end{align}
where $H'=H(s')$ and $\gamma'=(\rho',\epsilon')$.  For the first inequality above, the values $ \widehat{U}^{(\lambda)}$ are averages of the value for  $U^{(\lambda)} $.  The second inequality is for some $c>0$ by Lemma~\ref{LemHorseShoe}, and the third uses that $U^{(\lambda)}\big(s,f\big)$ is bounded by a multiple $c'$ of $\widehat{U}^{(\lambda)}\big(\gamma(s),f\big)$ for $s\in \Sigma$ with $l< H(s)$ by Part (2) of Lemma~\ref{ErrorEquation}.

Next, we bound the term (iii).     By Part (3) of Lemma~\ref{ErrorEquation},
\begin{align}\label{ComicBookMovie}
&\int_{\rho'\leq \lambda^{-1}   }  d\gamma'\,\big(1+ \textup{min}(\rho, \rho')\big)\,|\mathbf{E}_{\lambda}'(\gamma')| \nonumber  \\  & \leq C\int_{\rho'\leq \lambda^{-1}   }d\gamma'\,\big(1+ \textup{min}(\rho, \rho')\big)\,\int_{\Gamma_{V}}d\gamma''\,\widehat{T}_{\lambda}(\gamma',\gamma'')\,\mathcal{M}_{\lambda}(\gamma',\gamma'')\,\widehat{U}^{(\lambda)}\big(\gamma'',f\big) \nonumber\\ & \leq C\Big(\sup_{\rho''\geq \lambda^{-1} }    \widehat{U}^{(\lambda)}\big(\gamma'',f\big)   \Big)\int_{\rho'\leq \lambda^{-1}   }d\gamma'\,\big(1+ \textup{min}(\rho, \rho')\big)\int_{\rho''\geq \lambda^{-1} }d\gamma''\,\widehat{T}_{\lambda}(\gamma',\gamma'')\,\mathcal{M}_{\lambda}(\gamma',\gamma'')\nonumber \\ & \quad +C\int_{\rho''\leq \lambda^{-1} }d\gamma''\,\widehat{U}^{(\lambda)}\big(\gamma'',f\big) \,\int_{\rho'\leq \lambda^{-1}   }d\gamma' \big(1+ \textup{min}(\rho, \rho)\big)\,\widehat{T}_{\lambda}(\gamma',\gamma'')\,\mathcal{M}_{\lambda}(\gamma',\gamma'') \nonumber  \\ &  \leq  C'\lambda \sup_{\rho>\lambda^{-1}}\widehat{f}(\gamma)+C'\lambda^{2}\sup_{\rho\leq \lambda^{-1}}\widehat{U}^{(\lambda)}\big(\gamma,f\big) 
+C'\int_{\rho''\leq \lambda^{-1} }d\gamma''\,\frac{\big(1+\textup{min}( \rho'', \rho)\big)}{(1+\rho'')^{3}}\,\widehat{U}^{(\lambda)}\big(\gamma'',f\big).
\end{align}
For the  last inequality in~(\ref{ComicBookMovie}) follows from (I)-(III) below, where (I) and (II) are for the first term and (III) is for the second term.  
\begin{enumerate}[(I).]
\item There is  $c>0$ such that for all $\lambda<1$, 
$$ \sup_{\rho\geq \lambda^{-1} }    \widehat{U}^{(\lambda)}\big(\gamma,f\big)\leq  c\lambda^{-1} \sup_{\rho>\lambda^{-1}}f(\gamma)+c\sup_{\rho\leq \lambda^{-1}}\widehat{U}^{(\lambda)}\big(\gamma,f\big)+c\|f\|_{\infty}. $$
\item 
$$ \sup_{\lambda<1} \sup_{\rho,\rho''\leq \lambda^{-1}} \lambda^{-2}\int_{\rho'\leq \lambda^{-1}   }d\gamma'\,\big(1+ \textup{min}(\rho, \rho')\big)\int_{\rho''\geq \lambda^{-1} }d\gamma''\,\widehat{T}_{\lambda}(\gamma',\gamma'')\mathcal{M}_{\lambda}(\gamma',\gamma'')<\infty $$
\item 
$$ \sup_{\lambda<1}\sup_{\rho,\rho''\leq \lambda^{-1}} (1+\rho'')^{3}  \int_{\rho'\leq \lambda^{-1}   }d\gamma' \frac{\big(1+ \textup{min}(\rho, \rho')\big)}{ \big(1+ \textup{min}(\rho, \rho'')\big)    }\widehat{T}_{\lambda}(\gamma',\gamma'')\mathcal{M}_{\lambda}(\gamma',\gamma'') <\infty.  $$
\end{enumerate}
Statement (I) is from~(\ref{Rabbel}).   Statements (II) and (III) use the decay from $\mathcal{M}_{\lambda}(\gamma',\gamma'')$ and that the transition kernels $ \widehat{T}_{\lambda}(\gamma,\,\gamma')$ have uniformly bounded Gaussian tails in the quasi-momentum $|\mathbf{q}(\gamma)  -\mathbf{q}(\gamma')|$ for $\gamma=(\rho,\epsilon)$ with   $\rho\leq \lambda^{-1}$.  We do not go through the details of these inequalities.

\end{proof}

\section{Proof of Theorem~\ref{ThmMain}}\label{SecTheProof}

\vspace{.5cm}

\noindent [Proof of Theorem~\ref{ThmMain}]\vspace{.5cm}

 For $s\in \Sigma$ with $ H(s)>\frac{1}{2}\lambda^{-2}$, we have the inequality
\begin{align}\label{HighPee}
 U^{(\lambda)}\big(s,f\big) \leq  c\lambda^{-1}\sup_{H'>\frac{1}{2}\lambda^{-2}}f(s') +\sup_{H'\leq \frac{1}{2} \lambda^{-2}}U^{(\lambda)}\big(s',f\big) 
 \end{align}
for some $c>0$ by Lemma~\ref{LemHorseShoe}.  Thus, it is sufficient to prove the statement of the theorem for the domain $ H(s)\leq \frac{1}{2}\lambda^{-2}$.  

 For $s\in \Sigma$ with $l< H(s)\leq \frac{1}{2}\lambda^{-2}$, 
there is a $C>0$ such that for all $s$ and $\lambda<1$
$$U^{(\lambda)}\big(s,f\big)\leq C\widehat{U}^{(\lambda)}\big(\gamma(s),f\big)$$
 by Part (2) of Lemma~\ref{ErrorEquation}. Hence,  for the domain  $l< H(s)\leq \frac{1}{2}\lambda^{-2}$, it is sufficient to bound the values of  $\widehat{U}^{(\lambda)}\big(\gamma,f\big)$ for $\gamma=(\rho,\epsilon)$ with $\sqrt{2l}<\rho \leq\lambda^{-1}$.   By Lemma~\ref{CryBabyZero}, there are $c, c',L>0$ such that for all $\lambda<1$ and $\gamma=(\rho,\epsilon)$ with $\rho\leq \lambda^{-1}$, 
\begin{multline}\label{CryBabyUne}
\widehat{U}^{(\lambda)}\big(\gamma,f\big) \leq  c\|f\|_{\infty}+c\rho\sup_{H' >\frac{1}{2}\lambda^{-2} }f(s')+c\int_{\rho'\leq \lambda^{-1}    }d\gamma'\,\big(1+ \textup{min}(\rho', \rho)  \big)\,\widehat{f}(\gamma') \\+ c'\int_{ \sqrt{2L}\leq \rho'\leq \lambda^{-1}    }d\gamma'\,\widehat{U}^{(\lambda)}\big(\gamma',f\big)\,\frac{1+ \textup{min}(\rho, \rho')  }{ (1+ \rho')^{3}   },
\end{multline}
where  
\begin{align}\label{Italy}
 c'\int_{ \sqrt{2L}\leq \rho'\leq \lambda^{-1}    }d\gamma'\frac{1}{(1+\rho')^{2}}  \leq \frac{1}{2}.    \end{align}
It immediately follows that 
\begin{align}\label{Greece}
 \sup_{\rho\leq \lambda^{-1}}\widehat{U}^{(\lambda)} \big(\gamma, f\big) \leq  2c\|f\|_{\infty}+2c\lambda^{-1}\sup_{H'>\frac{1}{2}\lambda^{-2} }f(s')+2c\int_{\rho'\leq \lambda^{-1}}d\gamma'\,(1+ \rho')\widehat{f}(\gamma').
\end{align}
However, we need a more refined upper bound than~(\ref{Greece}) given by~(\ref{Gronwall}) below.
  By recursively applying the inequality~(\ref{CryBabyUne}) as in the proof of Gronwall's inequality, we obtain a series bound  
\begin{align}\label{Gronwall}
&\widehat{U}^{(\lambda)}\big(\gamma,f\big)\leq c\|f\|_{\infty}+c\rho\sup_{H'>\frac{1}{2}\lambda^{-2} }f(s')+ c\int_{\rho'\leq  \lambda^{-1}    }d\gamma'\,\big(1+ \textup{min}(\rho, \rho')  \big)\,\widehat{f}(\gamma')\nonumber \\
& + c\sum_{n=1}^{\infty}(c')^{n} \int_{\substack{\sqrt{2L}\leq \rho_{m}\leq \lambda^{-1}\\ 1\leq m\leq n   }  }d\gamma_{1}\cdots d\gamma_{n}\,\prod_{m=0}^{n-1} \frac{1+\textup{min}(\rho_{m+1}, \rho_{m})    }{ (1+\rho_{m+1}  )^{3} }\Big(\|f\|_{\infty}+\rho_{n}\sup_{H'>\frac{1}{2}\lambda^{-2} }f(s')    \Big) 
\nonumber \\ &+ c\sum_{n=1}^{\infty}(c')^{n}\int_{\rho'\leq \lambda^{-1}    }d\gamma'\,  \int_{\substack{\sqrt{2L}\leq \rho_{m}\leq \lambda^{-1}\\ 1\leq m\leq n   }  }d\gamma_{1}\cdots d\gamma_{n}\, \widehat{f}(\gamma')\big(1+\textup{min} (\rho',\rho_{n})  \big)\prod_{m=0}^{n-1} \frac{1+\textup{min}(\rho_{m+1}, \rho_{m} )   }{ (1+\rho_{m+1}  )^{3} }\nonumber  \\ &\leq  2c\|f\|_{\infty}+2c\rho\sup_{H'>\frac{1}{2}\lambda^{-2} }f(s')+ 2c\int_{\rho'\leq \lambda^{-1}    }d\gamma'\,\big(1+ \textup{min}(\rho', \rho)  \big)\,\widehat{f}(\gamma'),  
\end{align}
where we have denoted $\rho_{0}:=\rho$ in the argument of the  products. The second inequality uses that
\begin{align}\label{Tramp}
c'\int_{ \sqrt{2L}\leq \rho''\leq \lambda^{-1}    }d\gamma''\frac{\big(1+\textup{min}(\rho, \rho'') \big)\big(1+\textup{min}(\rho'', \rho' )\big)   }{(1+\rho'')^{3}}\leq \frac{1}{2}\big( 1+\textup{min}(\rho', \rho)\big),
\end{align}
which follows by~(\ref{Italy}).

  The series bound in~(\ref{Gronwall}) holds, since the error after the $n$th iteration of the inequality~(\ref{CryBabyUne}) is     
\begin{align*}
(c')^{n-1} \int_{\sqrt{2L}\leq \rho_{1},\dots, \rho_{n}\leq \lambda^{-1}}& d\gamma_{1}\cdots d\gamma_{n}\,\widehat{U}^{(\lambda)}\big(\gamma_{n},f\big)\frac{1+ \textup{min}(\rho_{1}, \rho ) }{ (1+\rho_{1}  )^{3} } \prod_{m=1}^{n-1} \frac{1+\textup{min}(\rho_{m+1}, \rho_{m})    }{ (1+\rho_{m+1}  )^{3} } \\  &\leq \frac{1}{2^{n-1}}\int_{\sqrt{2L}\leq \rho'\leq \lambda^{-1}}d\gamma'\, \widehat{U}^{(\lambda)}\big(\gamma',\,f\big)\frac{1+ \textup{min}(\rho' , \rho ) }{ (1+\rho'  )^{3} } \\ & \leq \frac{1}{2^{n-1}\sqrt{2L}} \sup_{\rho\leq \lambda^{-1}} \widehat{U}^{(\lambda)}\big(\gamma,\,f\big)\\ & \leq   \frac{c}{2^{n-2}\sqrt{2L}} \Big( \|f\|_{\infty}+\lambda^{-1}\sup_{H'>\frac{1}{2}\lambda^{-2} }f(s') +\int_{\rho'\leq \lambda^{-1}}d\gamma'\,(1+ \rho')\widehat{f}(\gamma')  \Big),
\end{align*}
which goes to zero for large $n$.  The first inequality is from~(\ref{Tramp}), and the third is~(\ref{Greece}).   Inequality~(\ref{Gronwall}) implies that there is a $c''>0$ such that for all $H(x,p)\geq l$,
\begin{multline}\label{GronwallII}
\widehat{U}^{(\lambda)}\big(\gamma(x,p),\,f\big) \\ \leq   c''\|f\|_{\infty}+c''|p|\sup_{H'>\frac{1}{2}\lambda^{-2}}f(s')   +c'' \int_{H'\leq \frac{1}{2} \lambda^{-2} }dp'\,dx'\,\big(1+\textup{min}(|p'|, |p|)\big)f(x',p'),   
\end{multline}
 where we have used that  $ \big| |p|-\rho(x,p)\big|\leq 2^{\frac{1}{2}}\big(\sup_{x}V(x)\big)^{\frac{1}{2}}$ is bounded since $2^{\frac{1}{2}}H^{\frac{1}{2}}(s)=\rho(s)$. This proves our bound for the domain $l< H(s)\leq \lambda^{-1}$.
  
Next, we bound $U^{(\lambda)}\big(s,f\big)$ in the domain $H(s)\leq l$. 
Let $\mathcal{T}_{\lambda}=U_{h'}^{(\lambda)} $ for $h'=1_{\Sigma}$.  The operator $\mathcal{T}_{\lambda}:B(\Sigma)\rightarrow B(\Sigma)$ is  the transition kernel for a Markov chain $\sigma_{n}$ (i.e. the resolvent chain).  We have the following  identity which is closely related to Part (3) of Proposition~\ref{LifeOperator}:    
\begin{multline*}
U^{(\lambda)}\big(s,f\big)=  \sum_{n=1}^{\infty}\int_{\substack{ m\leq n-1\\ H(s_{m})\leq l}  } \int_{H(s_{n})> l  } \mathcal{T}_{\lambda}(s,ds_{1})\cdots \mathcal{T}_{\lambda}(s_{n-1},ds_{n}) \\ \cdot\Big(\sum_{m=1}^{n-1} \big(1-h(s_{1})   \big)\cdots    \big(1-h(s_{m-1})   \big) \,f(s_{m})  +\big(1-h(s_{1})   \big)\cdots    \big(1-h(s_{n-1})   \big) \,  U^{(\lambda)}\big(s_{n},f\big) \Big).
\end{multline*}
The integration variable $s_{n}$ corresponds to the first time that the chain $\sigma_{n}$  jumps out of the set $H(s)\leq l$.   The above gives the inequalities
\begin{multline*}
U^{(\lambda)}\big(s,f\big)<  \sum_{n=1}^{\infty}\int_{\substack{ m\leq n-1\\ H(s_{m})\leq l}  } \int_{H(s_{n})> l  } \mathcal{T}_{\lambda}(s,ds_{1})\cdots \mathcal{T}_{\lambda}(s_{n-1},ds_{n})   \Big( n\|f\|_{\infty} +U^{(\lambda)}\big(s_{n},f\big)\Big)\\ \leq   \frac{ \sup_{\lambda<1}\sup_{H(s)\leq l}   \int_{H(s')\leq l  }\mathcal{T}_{\lambda}(s,ds')    }{\Big(1-\sup_{\lambda<1}\sup_{H(s)\leq l}   \int_{H(s')\leq l  }\mathcal{T}_{\lambda}(s,ds')  \Big)^{2}   }\Big(\|f\|_{\infty} +\sup_{\lambda<1}\sup_{H(s)\leq l}\int_{H(s')> l  } \mathcal{T}_{\lambda}(s,ds')\,U^{(\lambda)}\big(s',f\big)     \Big).  
\end{multline*}
 The second equality uses Holder's inequality, $U^{(\lambda)}\big(s_{n},f\big)\leq n U^{(\lambda)}\big(s_{n},f\big)$, and a geometric sum formula.    The probability $\int_{H(s')\leq l  }\mathcal{T}_{\lambda}(s,ds') $ can be easily shown to be bounded away from zero for all $\lambda<1$ and $s\in \Sigma$ by  considering the event that single a single collision occurs over the time interval $[0,\tau_{1}]$ and the particle jumps to an energy $> l$. Moreover, the jump measures $ \mathcal{T}_{\lambda}(s,ds')$ have uniformly bounded exponential tails for  $H(s)\leq l$ and $\lambda<1$ , since the collision rates $\mathcal{J}_{\lambda}(p,p')$ have Gaussian tails.  Thus, we can apply our bound for the values of $U^{(\lambda)}\big(s,f\big)$ over the domain $H(s')>l$ to obtain our required bound.

\section*{Acknowledgments}
This work is supported by the European Research Council grant No. 227772.

\end{document}